\documentclass[10pt,journal]{IEEEtran}

\usepackage{booktabs} % For formal tables
\usepackage{amsmath}
\usepackage{amsthm}
\usepackage{amsfonts}
\usepackage[english]{babel}
\usepackage{tikz}
\usepackage[utf8]{inputenc}
\usepackage{calc}
\usepackage{mathtools}

\usepackage{algorithm2e}
\usepackage{subcaption}

\usepackage[subject={Todo},author={Nicolas}]{pdfcomment}

\graphicspath{{}{figures/}}

 \newtheorem{theo}{Theorem}
\newtheorem{prop}{Proposition}
\newtheorem{corol}{Corollary}
\newtheorem{conj}{Conjecture}

\newcommand{\1}{\underline{1}}

\begin{document}

\title{Randomized Work Stealing versus Sharing in Large-scale Systems with Non-exponential Job Sizes}

\author{B. Van Houdt\\
Dept. Mathematics and Computer Science\\
University of Antwerp, Belgium}

%\keywords{}

\maketitle
\begin{abstract}
Work sharing and work stealing are two scheduling paradigms to redistribute work when performing distributed computations.
In work sharing, processors attempt to migrate pending jobs to other processors in the hope of reducing response times. 
In work stealing, on the other hand, underutilized processors attempt to steal jobs from other processors. 
Both paradigms generate a certain communication overhead and the question addressed in this paper is which
of the two reduces the response time the most given that they use the same amount of communication overhead.

Prior work presented explicit bounds, for large scale systems, on when randomized work sharing outperforms randomized work
stealing in case of Poisson arrivals and exponential job sizes and indicated that work sharing
is best when the load is below $\phi -1 \approx 0.6180$, with $\phi$ being the golden ratio. 

In this paper we revisit this
problem and study the impact of the job size distribution using a mean field model.
We present an efficient method to determine the boundary between the regions where sharing or stealing
is best for a given job size distribution, as well as bounds that apply to any (phase-type) job size distribution.
The main insight is that work stealing benefits significantly from having more variable job sizes and 
work sharing may become inferior to work stealing for loads as small as $1/2 + \epsilon$ for any $\epsilon > 0$.
\end{abstract}

\section{Introduction}

Work sharing and stealing are two fundamental scheduling paradigms to redistribute work in a 
distributed computing environment. The idea of work stealing is that any processor that becomes idle may 
attempt to steal a job from another processor with pending jobs. Work sharing on the other hand implies that
processors with pending jobs attempt to pass some of these jobs to idle processors. 
For instance, schedulers part of the Cilk programming language (developed at MIT in the 1990s), the Java
fork/join framework and the .NET Task Parallel Library implement work stealing. 

A particular class of work stealing and sharing strategies that has received considerable attention are the
so-called randomized work stealing/sharing strategies \cite{blumofe1,blumofe2,eager1,mitzenmacher1,mirchandaney1,gast1,minnebo2}. 
Under such a strategy  a processor that intends to
initiate a job transfer (either using stealing or sharing) probes another processor at random to see whether
a job can be transferred. Clearly, the more probes a processor uses, the more likely it becomes that a job can be transferred (between an idle
processor and a processor with pending jobs), which in turn reduces the mean response time. 

The main objective of this paper is to study whether work stealing or sharing achieves the lowest mean response time
in a large homogeneous system provided that both paradigms use the same average number of probe messages per time unit,
called the probe rate.
For Poisson arrivals and exponential job durations (with mean $1$) the following result was proven in \cite{minnebo2} using mean field models. 
As the system size tends to infinity and given that both paradigms use the same overall probe rate $r_{overall}$, 
sharing outperforms stealing  if and only if
\begin{align}\label{eq:testexp}
 \lambda < \frac{\sqrt{(r_{overall}+1)(r_{overall}+5)}-(r_{overall}+1)}{2},
\end{align}
in terms of the mean response time (as well as in the decay rate of the
queue length distribution). As $r_{overall}$ approaches zero, 
the right-hand side decreases to $\phi-1$, where $\phi=(1+\sqrt{5})/2$ is the golden ratio, which indicates
that work sharing prevails for any $r_{overall}$ when $\lambda < (-1+\sqrt{5})/2 \approx 0.6180$. 
In this paper we revisit this problem, but relax the assumption on exponential job sizes (by considering
phase-type distributions).   

Work stealing and sharing is mostly used in practice in the context of dynamic multithreading
where ongoing jobs spawn new jobs that are stored locally (and may be subsequently stolen or shared). They can however also be used 
in a context where all the jobs enter the system via one or multiple dispatchers to complement 
classic load balancing strategies such as the Join-the-Shortest-Queue among $d$ random choices (JSQd) \cite{vvedenskaya3,mitzenmacher2,bramsonLB,ying1}
or Join-the-Idle-Queue (JIQ) \cite{lu1,stolyar1}. The setting considered in this paper then corresponds to assuming new incoming jobs are
assigned among the processors in the system in a random manner. Future work in this
direction may exist in studying how these paradigms perform when combined with a more advanced load balancing algorithm
such as JSQd or JIQ. The model considered in this paper may also be applicable to a setting where a job is initially
assigned to a specific server for reasons such as data locality and can subsequently be migrated in case the
server is currently overloaded.

The modeling approach used in this paper exists in defining a mean field model, that is validated using simulation, and
studying the unique fixed point of this model to identify the region (in terms of $\lambda$ and $r_{overall}$) 
where work stealing/sharing is best for non-exponential job sizes. 
We further indicate that
by relying on the results in \cite{vanhoudt_global_PACM} one can formally
prove that the fixed point of the mean field model corresponds to the limit of the stationary distributions
of the finite systems provided that we truncate the queues and limit ourselves
to hyperexponential job size distributions (see Section \ref{sec:validate} for a brief discussion). %We do have some preliminary results in this direction and hope to settle this issue
%in a forthcoming publication.  

Below we highlight some of the main contributions. Contributions 2) to 6) are valid in the limit as the number of servers tends to infinity under the
assumption that the unique fixed point of the
mean field model is indeed the limit of the stationary distributions of the finite systems (as suggested
by simulation, see Section \ref{sec:validate}).
\begin{enumerate}
\item We present a mean field model for work stealing/sharing and prove that this model
has a unique fixed point that can be computed easily using matrix analytic methods. 
\item We devise a simple test to determine whether work stealing or sharing is best for
a given job size distribution, arrival and probe rate.
\item We prove that there exists a $\lambda^* \in [1/2,1)$ (that depends on the job size distribution
and probe rate) and an $r^*_{overall} \in (0,\lambda^2/(1-\lambda)]$ (that depends on the job size distribution
and arrival rate $\lambda$) such that stealing is best if and only if the arrival rate exceeds $\lambda^*$
and work sharing is best if and only if the probe rate exceeds $r^*_{overall}$.
\item We identify a region (in terms of the arrival and probe rate) 
where sharing is best  and present a conjecture
for the region where stealing is best for any phase-type job size distribution.
\item We show that work stealing benefits from having more variable job sizes and for
highly variable job size distributions (and probe rates below 1/2) stealing can outperform
work sharing as soon as the arrival rate exceeds 1/2. 
\item We derive explicit bounds for the case where the probe rate tends to zero. 
\end{enumerate}
The main insight is that when stealing and sharing use the same number of probes
per time unit, stealing benefits more from having more variable job sizes. The intuition behind this result goes as
follows: as the response time of a job reduces when a job is transferred to another server, transferring more
jobs should result in a larger reduction of the mean response time. If the overall probe rate is fixed, the number
of successful transfers is determined by the probability that a probe is successful. Under work sharing all
probes have a probability of $1-\lambda$ of being successful, irrespective of the job size distribution.
Under work stealing this probability is determined by the probability to have pending jobs. This
latter probability is expected to increase as the job sizes becomes more variable. So 
under sharing more job size variability has no impact on the number of successful probes per time unit,
while under stealing more variability yields more successes. 

The remainder of this paper is organized as follows.
In Section \ref{sec:system} we describe the system and work stealing/sharing strategies under consideration.
The mean field model is introduced in Section \ref{sec:meanfield} and a method to compute its unique fixed point 
is presented in Section \ref{sec:MPH1}. The mean field model is validated using simulation in Section \ref{sec:validate}. 
In Section \ref{sec:versus} we indicate how to determine whether stealing or sharing is best for a given job size
distribution, while
numerical results can be found in Section \ref{sec:num}. Bounds on the region where stealing/sharing is best
for any (phase-type) distribution are derived in Section \ref{sec:bounds}. Finally, explicit results for these
bounds for sufficiently small probe rates are established in Section \ref{sec:small} and conclusions are
drawn in Section \ref{sec:concl}. 

\section{System description and strategies}\label{sec:system}

The system analyzed in this paper has the following characteristics:
\begin{enumerate}
\item  The system consists of $N$ homogeneous servers that process incoming jobs in FCFS order and each server
has an infinite buffer to store jobs. 
\item Each server is subject to its own local Poisson arrival process with rate $\lambda$. 
\item  The time required to transfer probe messages and jobs 
between different servers can be neglected in comparison with the processing time, i.e., 
the probe and job transfers are assumed to be instantaneous. 
\end{enumerate}

This setting is similar to the one considered in \cite{eager1,mitzenmacher1,mirchandaney1,gast1,minnebo2}, the main difference is that 
the job sizes are not assumed to be exponential. 
Instead we assume the job sizes follow a continuous time phase-type distribution with mean $1$ characterized by the $n \times n$ 
subgenerator\footnote{$S$ is a subgenerator matrix if its diagonal entries are negative, its off-diagonal entries are non-negative and
its row sums are negative.} matrix $S$
and initial vector $\alpha=(\alpha_1,\ldots,\alpha_n)$. Its cumulative distribution function (cdf) $H$ and probability density function (pdf) 
$h$ is given by $H(y) = 1- \alpha e^{Sy} \1$ and $h(y) = \alpha e^{Sy} \mu$, respectively, where $\1$ is a vector of ones and $\mu  = -S \1$. 
Note that $\alpha_i$ is the probability that a job starts service in phase $i$, entry $(i,j)$ of $S$, 
for $i \not= j$, contains the rate at which the job in service changes its service phase from $i$ to $j$ and $\mu_i$ is the rate
at which a job in phase $i$ completes service. For example whenever $S$ is a diagonal matrix, 
the phase type distribution is a hyperexponential distribution.

We note that any general distribution on $[0,\infty)$ 
can be approximated arbitrary closely with a phase-type distributions \cite{tijms1994stochastic} and various fitting tools to do so are available online,
e.g., \cite{jPhase,ProFiDo}. In addition, we believe that many of the stealing/sharing bounds presented in this paper are valid for any
job size distributions. In fact, some of the arguments (based on coupling) presented in the paper are not limited to phase-type distributions.

We consider the following work sharing and stealing mechanisms, called the push and pull strategy in \cite{minnebo2}:
\begin{enumerate}
\item {\it Sharing:} Whenever a server has $\ell \geq 2$ jobs in its queue,
meaning $\ell-1$ jobs are waiting to be served, the server generates 
probe messages at rate $r$. Thus, as long as the number of jobs in
the queue remains above $1$, probes are sent according to a 
Poisson process with rate $r$. Whenever the queue length $\ell$ drops to $1$,
this process is interrupted and remains interrupted as long as the
queue length is below $2$. The server that is probed is selected at
random and is only allowed to accept a job if it is idle.
\item {\it Stealing:} Whenever a server has $\ell = 0$ jobs in its queue,
meaning the server is idle, it generates 
probe messages at rate $r$. Thus, as long as the server remains idle, 
probes are sent according to a Poisson process with rate $r$. 
This process is interrupted whenever the server becomes busy.
The probed server is also selected at
random and a probe is successful if there are jobs {\it waiting}
to be served. Thus, a job in service is never stolen by another server.
\end{enumerate}
The overall probe rate $r_{overall}$ therefore equals $r$ times the probability
that a server holds two or more jobs for work sharing and $r$ times the
probability that a server is idle for work stealing.

The work stealing and sharing strategies  considered in \cite{eager1,mirchandaney1,mirchandaney2}
operate as follows. The work stealing strategy tries to attract a job whenever
the server becomes idle, while the work sharing strategy tries to
get rid of arriving jobs if the server is busy upon their arrival. Further, instead of sending
a single probe, both strategies repeatedly send probes until either one gets a positive
reply or a predefined maximum of $L_p$ probes is reached. 
The overall probe rate $r_{overall}$ of these more traditional strategies clearly depends on $L_p$ and the load $\lambda$, 
which makes it hard to compare these  strategies in a fair manner.  
In \cite{minnebo2} it was shown that under exponential job sizes, 
the fixed points of the mean field models of these traditional strategies and the above strategies that make use of the probe rate $r$ 
coincide if $r$ is set such that the overall probe rate $r_{overall}$ of these more traditional strategies is matched.
As such the traditional strategies do not offer a performance benefit compared to the ones considered in this paper
when the system becomes large and job sizes are exponential.

In the next section we introduce a {\it single} mean field model that is intended to
capture the behavior of the system as $N$ tends to infinity for sharing {\it and} stealing. 
A single mean field model can be used as for both strategies the rate at which jobs are
transferred between servers is given by $Nr$ times the fraction of idle servers
times the fraction of servers with pending jobs. 
However, as the probability that a server is idle does in general not match the probability that it has pending jobs,
the overall probe rate $r_{overall}$ typically differs when both strategies
rely on the same $r$. Hence, as in \cite{minnebo2} which was limited to exponential job sizes, 
we aim at comparing these strategies when the rates $r$ are set such that the overall probe
rate $r_{overall}$ matches some predefined probe rate. 
We do {\it not} consider hybrid strategies where servers probe at some rate $r_1$ when being idle and at
some rate $r_2$ when they have pending jobs. 

\section{Mean field model}\label{sec:meanfield}
 For $\ell > 0$, denote $f_{\ell,j}(t)$ as the fraction of queues with length $\ell$ at time $t$, the server of which is in phase $j$.
Let $f_0(t)$ be the fraction of idle queues at time $t$. Let $1[A]$ be equal to one if $A$ is true and to zero otherwise. 
We propose to use the following ODE model:
\begin{align*}
\frac{d}{dt}&f_{\ell,j}(t) = \lambda f_{\ell-1,j}(t) 1[\ell > 1] - \lambda f_{\ell,j}(t) + \lambda f_0(t) \alpha_j 1[\ell = 1]\\
& + \left( \sum_{j'} \mu_{j'} f_{\ell+1,j'}(t) \right)\alpha_j 
-f_{\ell,j}(t) \mu_j + \sum_{j'\not= j} f_{\ell,j'}(t) s_{j',j} \\
& -f_{\ell,j}(t)  \sum_{j'\not= j} s_{j,j'} + f_0(t) r (f_{\ell+1,j}(t)  - 1[\ell > 1] f_{\ell,j}(t))   \\
& + 1[\ell = 1] rf_0(t) \left( 1-f_0(t)-\sum_{j'} f_{1,j'}(t)\right) \alpha_j,  
\end{align*}
for $\ell \geq 1$ and
\begin{align*}
\frac{d}{dt}f_0(t) &= -\lambda f_0(t) + \sum_{j'} \mu_{j'} f_{1,j'}(t) \\
&-  r f_0(t) \left( 1-f_0(t)-\sum_{j'} f_{1,j'}(t)\right).
\end{align*}
The first three terms for the drift of $f_{\ell,j}(t)$ correspond to arrivals, the next two to service completions,
the two sums to phase changes and the latter two terms are caused by job transfers. 
There are two ways to think about the terms regarding the job transfers. If we consider
work stealing $r f_0(t)$ is the rate at which idle servers attempt to steal work. 
With probability $f_{\ell+1,j}(t)$ they probe a server in phase $j$ with length $\ell +1$ and
thus create an additional server with $\ell$ jobs in phase $j$. With probability $f_{\ell,j}(t)$ a probe steals
a job from a server with length $\ell$ in phase $j$, lowering the number of servers of this type,
unless $\ell = 1$, in which case there is no steal. The last term corresponds to the
increase in servers with one job in phase $j$, such a server is created if a server with
two or more jobs is probed and the initial phase of service equals $j$. Alternatively we could think
in terms of the sharing strategy, $rf_{\ell+1,j}(t)$ is now the rate at which servers with $\ell+1$
jobs in phase $j$ are trying to transfer one of their pending jobs and $f_0(t)$ is the probability
that such a server succeeds. The other two terms can be interpreted in the same way.
For the drift of $f_0(t)$
the first term is due to arrivals, the second due to service completions and the last one is due to job transfers.

We can write these equations in matrix form as follows. Let $\vec f_\ell(t) = (f_{\ell,1}(t),\ldots,$ $f_{\ell,n}(t))$, 
then as $\mu_j = -\sum_{j'} s_{j,j'}$ the above set of ODEs can be expressed as
\begin{align}\label{eq:ODEl>0}
\frac{d}{dt}\vec f_{\ell}(t) &= \lambda \vec f_{\ell-1}(t)  1[\ell > 1]- \lambda \vec f_{\ell}(t) + \lambda f_0(t) \alpha 1[\ell = 1] \nonumber \\
& + \vec f_{\ell+1}(t) \mu \alpha + f_0(t) r (\vec f_{\ell+1}(t) - 1[\ell > 1] \vec f_{\ell}(t))  \nonumber \\  & +\vec f_\ell(t) S  + 1[\ell = 1] r f_0(t) \left( 1-f_0(t)-\vec f_{1}(t) \1\right) \alpha,  
\end{align}
for $\ell \geq 1$ and
\begin{align}\label{eq:ODEl=0}
\frac{d}{dt}f_0(t) &= -\lambda f_0(t) + \vec f_{1}(t) \mu -  r f_0(t) \left( 1-f_0(t)-\vec f_{1}(t)\1 \right).
\end{align}
In the next section we show  this set of ODEs has a unique fixed point $\zeta$ 
(with $\zeta_0 + \sum_{\ell \geq 1} \vec \zeta_\ell \1 =1$) that can be expressed as
the invariant distribution of an M/PH/1 queueing system with negative customers.  
The next proposition is used to establish this result. 
Let $\beta$ be the unique stochastic vector\footnote{Note that $S+\mu \alpha$ is the rate matrix of
an $n$ state continuous time Markov chain and $\beta$ is its unique steady state vector.} such that $\beta (S+\mu\alpha)=0$. 
Entry $i$ of the vector $\beta$ is the probability that the server is in phase $i$ if we observe a busy server
at a random point in time. As the mean service time
is assumed to be $1$, we have $\beta \mu = 1$. 

\begin{prop}\label{prop:zeta}
For any fixed point  $\zeta = (\zeta_0,\vec \zeta_1, \vec \zeta_2, \ldots)$ with 
$\zeta_0 + \sum_{\ell \geq 1} \vec \zeta_\ell \1 =1$
of the set of ODEs (\ref{eq:ODEl>0}-\ref{eq:ODEl=0}), we have
$\zeta_0 = 1-\lambda$ and $\sum_{\ell \geq 1} \vec \zeta_\ell = \lambda \beta$.
\end{prop}
\begin{proof}
By noting that 
\[\sum_{\ell \geq 1} \frac{d}{dt}\vec f_{\ell}(t) + \frac{d}{dt}f_0(t) \alpha = \left(\sum_{\ell \geq 1} \vec f_{\ell}(t)\right) (S+\mu \alpha),\]
we find that $\left( \sum_{\ell \geq 1} \vec \zeta_\ell \right) (\mu \alpha + S) = 0$. 
Therefore $\sum_{\ell \geq 1} \vec \zeta_\ell$ is proportional to $\beta$, the unique stochastic vector
such that $\beta (\mu \alpha + S) = 0$. 
Furthermore, 
\begin{align*}
\sum_{\ell \geq 1} \ell  \frac{d}{dt}\vec f_{\ell}(t) &= \lambda \sum_{\ell \geq 1} \vec f_\ell(t) + \lambda f_0(t) \alpha + \sum_{\ell \geq 1}
\ell \vec f_{\ell+1}(t) \mu \alpha \\ & +\sum_{\ell \geq 1} \ell \vec f_{\ell}(t) S - r f_0(t) \sum_{\ell \geq 2} \vec f_{\ell}(t) 
\\& + r f_0(t) \left( 1-f_0(t) - \vec f_1(t)\right) \alpha,
\end{align*}
where 
\begin{align*}
\sum_{\ell \geq 1} \ell \vec f_{\ell+1}(t) \mu \alpha &+  \sum_{\ell \geq 1} \ell \vec f_{\ell}(t) S  = 
\vec f_1(t)S \\&+ \sum_{\ell \geq 2} \ell \vec f_{\ell}(t) (S+\mu \alpha) -\sum_{\ell \geq 2} \vec f_{\ell}(t) \mu \alpha.
\end{align*} 
As $\alpha \1 = 1$, $f_0(t)+\sum_{\ell \geq 1} \vec f_\ell(t) \1 = 1$ and $(S+\mu \alpha)\1 = 0$, we find
\begin{align*}
\sum_{\ell \geq 1} \ell \frac{d}{dt}\vec f_{\ell}(t) \1 &= \lambda -  \sum_{\ell \geq 1} \vec f_{\ell}(t) \mu.
\end{align*}
This implies that $\left( \sum_{\ell \geq 1} \vec \zeta_\ell \right) \mu = \lambda$,
which yields $\left( \sum_{\ell \geq 1} \vec \zeta_\ell \right) = \beta \lambda$ as 
$\beta \mu = 1$.
\end{proof}

\section{M/PH/1 queue with negative customers}\label{sec:MPH1}
We now introduce an M/PH/1 queueing system with negative customers and show that its steady state distribution corresponds to the unique fixed point
of the set of ODEs given by (\ref{eq:ODEl>0}-\ref{eq:ODEl=0}). The queueing system has the following characteristics:
\begin{enumerate}
\item Arrivals occur according to a Poisson process with rate $\lambda$ when the server is busy and at rate
$\lambda_0$ when the server is idle.
\item There is a single server, infinite waiting room and service times follow a phase-type distribution $(\alpha,S)$ with mean $1$.
Customers are served in FCFS order.
\item Negative arrivals occur at rate $(1-\lambda)r$ when the queue length {\it exceeds} one and reduce the queue length
by one (by removing a customer from the back of the queue).
\item The arrival rate $\lambda_0$ is such that the probability of having an idle queue is $1-\lambda$ and
thus depends on $\lambda, r$ and $(\alpha,S)$ only. 
\end{enumerate}

We start by defining a continuous time Markov chain
$(X_t(r),$ $Y_t(r))_{t \geq 0}$ where $X_t(r)$ denotes the number of jobs in the queue at time $t$ and $Y_t(r) \in {1,\ldots,n}$ is the server
phase at time $t$ provided that $X_t(r) > 0$. Define
\begin{align*}
A_{-1}(r) &= \mu \alpha + (1-\lambda) r I,\\ 
A_0(r) &= S - (\lambda + (1-\lambda) r) I,\\
A_1 &= \lambda I.
\end{align*}
Using these matrices we define the rate matrix $Q(r)$ of the Markov chain $(X_t(r),Y_t(r))_{t \geq 0}$  as 
\begin{align}\label{eq:Q}
Q(r)=\begin{bmatrix}
-\lambda_0(r) & \lambda_0(r) \alpha & & & \\
\mu & S-\lambda I  & A_1 & &\\
& A_{-1}(r) & A_0(r) & A_1 & \\
& & \ddots & \ddots & \ddots 
\end{bmatrix}.
\end{align}
Note the matrix $Q(r)$ is fully determined by $\lambda, r, \alpha$ and $S$, except for the rate $\lambda_0(r)$ to exit
level $0$.  We define $\lambda_0(r)$ further on.

Let $\pi_0(r) = \lim_{t \rightarrow \infty} P[X_t(r) = 0]$, $\pi_{\ell,j}(r) = \lim_{t \rightarrow \infty} P[X_t(r) = \ell, Y_t(r) = j]$ for $\ell > 0$ and
$\pi_\ell(r) = (\pi_{\ell,1}(r),\ldots,\pi_{\ell,n}(r))$. Finally, set $\pi_{k+}(r)=\sum_{\ell \geq k} \pi_\ell(r)$. 
Due to the Quasi-Birth-Death (QBD) structure \cite{neuts2} we have
\begin{align*}
\pi_1(r) &= \pi_0(r) R_1(r),\\
\pi_\ell(r) &= \pi_1(r) R(r)^{\ell-1},
\end{align*}
for $\ell > 1$, where the $n \times n$ matrix $R(r)$ is the smallest nonnegative solution to
\begin{align*}
A_1 + R(r)A_0(r) + R(r)^2 A_{-1}(r) = 0
\end{align*}
and
\begin{align}\label{eq:nonR1}
\lambda_0(r) \alpha + R_1(r) (S-\lambda I) + R_1(r) R(r) A_{-1}(r) = 0.
\end{align}
As $A_1 G(r) = R(r) A_{-1}(r)$ (see \cite{neuts2}), where $G(r)$ is the smallest nonnegative solution to $A_{-1}(r)+A_0(r)G(r)+A_1G(r)^2=0$, we find
\begin{align}\label{eq:R1}
R_1(r) = -\lambda_0(r) \alpha (S+\lambda(G(r)-I))^{-1}.
\end{align}
The matrix $G(r)$ is a stochastic matrix whenever the Markov chain characterized by $Q(r)$
is positive recurrent (which clearly holds for any $r \geq 0$ if $\lambda < 1$) \cite{neuts2}.
Therefore the matrix $G(r)-I$ is a generator matrix,
$S+\lambda(G(r)-I)$ is a subgenerator matrix and the latter is therefore invertible.
Note that $R(r)$ and $G(r)$ are independent of $\lambda_0(r)$.

If we observe the Markov chain only when $X_t(r) > 0$ and
only focus on $Y_t(r)$, we find that it evolves according to the generator
matrix $A_{-1}(r)+A_0(r)+A_1 = S + \mu \alpha$. Hence, the vector $\beta$ 
represents the stationary distribution of the service phase given that the
server is busy. As such we find 
\begin{align}\label{eq:sumpis}
\sum_{\ell > 0} \pi_\ell(r) = (1-\pi_0(r)) \beta.
\end{align}
 
\paragraph{Defining $\lambda_0(r)$} To define $\lambda_0(r)$ we demand that
$\pi_0(r) = 1- \lambda$ and that $\pi_0(r)+\sum_{\ell \geq 1} \pi_\ell(r) \1 = 1$. This implies
\begin{align*}
1 = \pi_0(r)+\sum_{\ell \geq 1} \pi_\ell(r) \1 = (1-\lambda) (1+R_1(r) (I-R(r))^{-1} \1). 
\end{align*}
Equation \eqref{eq:R1} yields
\begin{align}\label{eq:lam0}
\lambda_0(r) = \frac{\lambda}{(1-\lambda)\alpha (\lambda(I-G(r))-S)^{-1}(I-R(r))^{-1}\1}.
\end{align}
As $\lambda, r, \alpha$ and $S$ fully determine the matrices $R(r)$ and $G(r)$, they also
fully determine $\lambda_0(r)$ using the above equation. 

Although we have a definition for $\lambda_0(r)$, we now derive a second equivalent expression.
To do so, note that the rate at which the level (being $X_t(r)$) of the QBD goes up should be matched by the rate that the level goes down.
Hence, we have
\[ \pi_0(r) \lambda_0(r) + (1-\pi_0(r)) \lambda = \sum_{\ell > 0} \pi_\ell(r) \mu + \sum_{\ell > 1} \pi_\ell(r) \1 (1-\lambda)r. \]
If we demand that $\pi_0(r) = (1-\lambda)$, \eqref{eq:sumpis} and
the equality $\beta \mu = 1$ implies
\[ (1-\lambda)\lambda_0(r)  + \lambda^2 = \lambda + \pi_{2+}(r)\1 (1-\lambda)r, \]
which simplifies to 
\begin{align}\label{eq:lam0v2} 
\lambda_0(r)  = \lambda + r \pi_{2+}(r)\1. 
\end{align}
Using \eqref{eq:R1} and the equality $\pi_{2+}(r) \1= \lambda - \pi_1(r)\1$, this provides us with the following equivalent definition for $\lambda_0(r)$
\begin{align}\label{eq:lam0v3}
\lambda_0(r) = \frac{\lambda (1+r)}{1+(1-\lambda)r\alpha (\lambda(I-G(r))-S)^{-1}\1}.
\end{align}

\paragraph{Steady state probabilities}
Using  \eqref{eq:lam0} we see that the steady state probabilities $\pi_\ell(r)$ can be expressed as
\begin{align}\label{eq:steady}
\pi_\ell(r) = \lambda \frac{\alpha (\lambda(I-G(r))-S)^{-1}R(r)^{\ell-1}}{\alpha (\lambda(I-G(r))-S)^{-1}(I-R(r))^{-1}\1},
\end{align}
and $\pi_0(r) = 1-\lambda$.
\begin{theo}\label{thm:fixed}
The steady state probability vector given by \eqref{eq:steady} is 
the unique fixed point $\zeta$ of the set of ODEs given by (\ref{eq:ODEl>0}-\ref{eq:ODEl=0})
with $\zeta_0 + \sum_{\ell \geq 1} \vec \zeta_\ell \1 =1$.
\end{theo}
\begin{proof}
We show that any fixed point $\zeta= (\zeta_0,\vec \zeta_1, \vec \zeta_2, \ldots)$ of 
(\ref{eq:ODEl>0}-\ref{eq:ODEl=0}) satisfies $\zeta Q(r) = 0$. The result then follows from
the uniqueness of the stationary distribution of the Markov chain.

Due to Proposition \ref{prop:zeta} we have $\zeta_0 = 1-\lambda$ and it suffices to show that
$(\vec \zeta_1, \vec \zeta_2, \ldots)$ is an invariant vector of
\begin{align}\label{eq:Qp}
Q_{busy}(r)=\begin{bmatrix}
S-\lambda I+\mu \alpha  & A_1 & \\
 A_{-1}(r) & A_0(r) & \ddots  \\
 & \ddots & \ddots  
\end{bmatrix},
\end{align}
being the rate matrix of the chain censored on the states where the server is busy.

Due to Proposition \ref{prop:zeta}, Equation \eqref{eq:ODEl>0} with $\ell > 1$, is equivalent to
\[ 0 = \lambda \vec \zeta_{\ell-1} + \vec \zeta_\ell(S-\lambda I - (1-\lambda)r I) + \vec \zeta_{\ell+1}(\mu \alpha + (1-\lambda)r I),\]
which means $0 = \vec \zeta_{\ell-1}A_1  + \vec \zeta_\ell A_0(r)  + \vec \zeta_{\ell+1}A_{-1}(r)$.
Similarly, Equation  \eqref{eq:ODEl>0} with $\ell = 1$ and Equation  \eqref{eq:ODEl=0} can be written as
\begin{align*} 
0 &= \zeta_0 (\lambda + r(1-\zeta_0-\vec \zeta_1 \1))\alpha + \vec \zeta_1 (S-\lambda I) + \vec \zeta_2 A_{-1}(r),\\
0 &= -\zeta_0 (\lambda + r(1-\zeta_0-\vec \zeta_1 \1)) + \vec \zeta_1 \mu.
\end{align*}
When combined this yields $\vec \zeta_1 (S-\lambda I+\mu \alpha) + \vec \zeta_2 A_{-1}(r) = 0$.
This is the first balance equation of $(\vec \zeta_1, \vec \zeta_2, \ldots) Q_{busy}(r) = 0$.
%Plugging \eqref{eq:steady}, with $\ell > 1$, in the fixed point equation \eqref{eq:ODEl>0} implies
%\begin{align*}
%0 = \pi_{\ell-1}(r) \left(\lambda (I-R(r)) + R(r)^2 \mu \alpha + (1-\lambda)r (R(r)^2-R(r))+R(r)S \right),
%\end{align*}
%which is equivalent to demanding that $R(r)^2 A_{-1}(r) + R(r) A_0(r) + A_1 = 0$. For $\ell = 1$ we obtain
%\begin{align*}
%0 &= \pi_0(r) \left[\lambda (\alpha-R_1(r)) + R_1(r) R(r) \mu \alpha + (1-\lambda)r R_1(r) R(r) \right.\\
%& \hspace*{2cm} \left.+R_1(r) S + r(1-\pi_0-\pi_1(r) e) \alpha \right] \\
% &= \pi_0(r) \left[ R_1(r) R(r) A_{-1}(r) + R_1(r) (S-\lambda I) + (\lambda + r(1-\pi_0-\pi_1(r) e))\alpha \right].
%\end{align*}
%Due to \eqref{eq:nonR1} this is equivalent to demanding that $\lambda_0(r) = \lambda + r(1-\pi_0-\pi_1(r) e)$, which
%was proven earlier, see \eqref{eq:lam0v2}.
%Finally, the fixed point equation for $\ell = 0$ is given by
%\begin{align*}
%0 &= \pi_0(r) \left[ -\lambda  + R_1(r) \mu - r(1-\pi_0-\pi_1(r)e) \right] = \pi_0 \left[ R_1(r) \mu -\lambda_0(r) \right],
%\end{align*}
%which follows from the first balance equation of $\pi(r) Q(r)=0$.
\end{proof}

\subsection{Special cases}

In general there does not appear to exist an explicit formula for $\pi_\ell(r)$ as $G(r)$ and $R(r)$ are the solutions
to quadratic matrix equations. Note that these matrices and thus $\pi_\ell(r)$ can be computed in a fraction of a second
using matrix analytic methods (see \cite{bini3}). In this subsection we consider two special cases for which we can obtain an explicit result for
$G(r)$.  

\paragraph{Exponential job durations} In case of exponential job lengths with mean $1$, we clearly have $G(r) =1$. Therefore 
$R(r) = \lambda/(1+(1-\lambda)r)$ as $\lambda G(r) = R(r)A_{-1}(r)$ and $\pi_\ell(r)$  simplifies to
\begin{align*}
\pi_\ell(r) = \lambda \left(\frac{\lambda}{1+(1-\lambda)r}\right)^{\ell-1}\left/ \left(1-\frac{\lambda}{1+(1-\lambda)r}\right)^{-1} \right.,
\end{align*}
hence $\sum_{\ell \geq i} \pi_\ell(r)  = \lambda^i /(1+(1-\lambda)r)^{i-1}$, which is the unique fixed point of the
set of ODEs for the exponential job durations derived in \cite{minnebo2}.
 
\paragraph{Hyperexponential job durations} In case of hyperexponential job sizes with $2$ phases,
that is, when $S=diag(-\mu_1,-\mu_2)$ and $\alpha = (p, 1-p)$, we can write $G(r)$ as 
\[G(r) = \begin{bmatrix} 1-g_{1,2}(r) & g_{1,2}(r) \\
g_{2,1}(r) & 1-g_{2,1}(r)
\end{bmatrix},
\]
and show that $g_{2,1}(r)$ is a root of the perturbed polynomial 
\[ax^3+(b+(\mu_2-\mu_1)(1-\lambda)r)x^2+(c-p \mu_2(1-\lambda)r)x+d,\]
with
\begin{align*}
a&=\lambda(\mu_2-\mu_1), \\
b&=\lambda(\mu_1-\mu_2)p - (\mu_1-\mu_2) \mu_2 -\lambda \mu_2,\\
c&=(\mu_1-\mu_2) \mu_2 p + \lambda \mu_2 p - \mu_2^2 p, \\
d& = \mu_2^2 p^2 ,
\end{align*}
where $p$ is a root of $ax^3+bx^2+cx+d$.
The probability $g_{1,2}(r)$ is also a root of a degree $3$ polynomial where the
coefficients are the same as for $g_{2,1}(r)$ if we replace $p$ by $1-p$ and exchange $\mu_1$ and $\mu_2$.
Thus, it is possible to express the entries of $G(r)$ explicitly in terms of $\mu_1, \mu_2, p$ and $\lambda$, but
the expressions look very involved.

\begin{table}[t!]
\center
\begin{tabular}{c c c c c c}
 \hline
$\lambda$ & $SCV$ & $r$ & ODE & simul $\pm$ conf & rel. error \\
\hline 
1/2 & 1/5 & 1/5 & 1.5276  &  1.5279$\pm$0.0021 &   0.0002  \\
 &  & 1 &       1.3644  &  1.3640$\pm$0.0021 &   0.0003  \\
 &  & 5 &     1.1514  &  1.1516$\pm$0.0018 &   0.0002  \\
\hline
1/2 & 5 & 1/5 &      3.2285  &  3.2312$\pm$0.0078 &   0.0008  \\
 &  & 1 &      1.8847  &  1.8844$\pm$0.0036 &   0.0002  \\
 &  & 5 &      1.1795  &  1.1782$\pm$0.0022 &   0.0011  \\
\hline
1/2 & 25 & 1/5   &     9.9397  &  9.9612$\pm$0.0438 &   0.0022  \\
 &  & 1  &    2.8406  &  2.8487$\pm$0.0159 &   0.0028  \\
 &  & 5  &     1.1855  &  1.1854$\pm$0.0028 &   0.0001 \\
\hline \hline
3/4 & 1/5 & 1/5  &2.5451  &  2.5461$\pm$0.0022   & 0.0004  \\
 &  & 1 &   2.0148  &  2.0163$\pm$0.0026  &  0.0007  \\
  & & 5 &   1.4142  &  1.4137$\pm$0.0015  &  0.0004  \\
\hline
3/4  & 5 & 1/5 &   8.1885  &  8.1850$\pm$0.0263  &  0.0004  \\
  &  & 1 &   4.6246  &  4.6311$\pm$0.0127  &  0.0014 \\
  &  & 5 &   1.6517  &  1.6594$\pm$0.0028  &  0.0047 \\
\hline
3/4  & 25 & 1/5 &  31.5323  & 31.6262$\pm$0.2779  &  0.0030  \\
  &  & 1 &  14.7129  & 14.7340$\pm$0.0944  &  0.0014  \\
  &  & 5 &   1.8058  &  1.8283$\pm$0.0091  &  0.0123  \\
\hline \hline
7/8 & 1/5 & 1/5 & 4.5552  &  4.5487$\pm$0.0121   & 0.0014  \\
 & & 1 &   3.2503   & 3.2483$\pm$0.0058  &  0.0006  \\
  & & 5 &  1.8774  &  1.8788$\pm$0.0024  &  0.0007  \\
\hline 
7/8 & 5 & 1/5 &   18.1843  & 18.2002$\pm$0.1004 &   0.0004  \\
& & 1 &   10.5504  & 10.5452$\pm$0.0694   & 0.0005  \\
& & 5 &    3.1684 &   3.1927$\pm$0.0132 &   0.0076  \\
\hline
7/8 & 25 & 1/5 &   74.8468 &  74.0279$\pm$0.8398 &   0.0111  \\
& & 1 &   40.5751  & 40.5135$\pm$0.4818 &   0.0015  \\
& & 5 &    6.9646  &  7.1031$\pm$0.0787 &   0.0195  \\
\hline
\end{tabular}
\caption{Mean job response time: ODE model versus simulation for $N=1000$ servers.}
\label{tab:validate}  
\end{table}

\section{Model validation}\label{sec:validate}

In this section we validate the ODE model by considering both distributions with a squared coefficient of variation (SCV) smaller and larger than one.
For the case with $SCV > 1$, we make use of the class of hyperexponential distributions with $2$ phases with parameters 
$(\alpha_1,\mu_1,\mu_2)$. This means that with probability $\alpha_i$
a job is a type-$i$ job and has an exponential duration with mean $1/\mu_i$, for $i=1,2$ (where $\alpha_2 = 1 - \alpha_1)$. 
Apart from matching the mean (to $1$) and the SCV we match the fraction $f$ of the workload that is contributed by the type-$1$ jobs 
(i.e., $f = \alpha_1/\mu_1$). In case $\mu_1 \gg \mu_2$ this can be interpreted as stating that a fraction $f$ of the workload is contributed by the {\it short} jobs.
The mean (being $1$), $SCV$ and fraction $f$ can be matched as follows:
\begin{align}\label{eq:HyperExp}
 \mu_{1} &= \frac{SCV+(4f-1)+\sqrt{(SCV-1)(SCV-1+8f\bar f)}}{2f(SCV+1)},\\
 \mu_{2} &= \frac{SCV+(4\bar f-1)-\sqrt{(SCV-1)(SCV-1+8f\bar f)}}{2\bar f(SCV+1)}, 
\end{align}
with $\bar f = 1-f$  and $\alpha_{1} = \mu_1 f$. Table \ref{tab:HEXP} lists the parameter settings
for the distributions considered in the plots.
For the case with $SCV \leq 1$, we consider the class of hypoexponential distributions (i.e., distributions that are the sum of $k \geq 1$
independent exponential random variables) and limit ourselves for the most part to Erlang distributions.  

In Table \ref{tab:validate}  we validate the ODE model by comparing the mean response time obtained from
the unique invariant distribution of the M/PH/1 queue with negative arrivals with simulation
experiments for $N=1000$ servers. The simulation started from an empty system and  the
system was simulated up to time $t=5000$ (for $SCV=1/5$ and $5$) or $t=25000$ (for $SCV = 25$)
with a warm-up period of $33\%$. The $95\%$ confidence intervals were computed based on $20$
runs. We observe that the relative error of the mean field model tends to increase with $\lambda$ and the SCV,
but remains below $2\%$ for the $27$ (arbitrary) cases considered and is even below $0.1\%$ in quite a few cases.

Let $\pi_\ell^{(N)}(r)$ be the steady state probability that an arbitrary server contains exactly $\ell$ jobs in a system
consisting of $N$ servers with probe rate $r$. We believe that the probabilities $\pi_\ell^{(N)}(r)$ converge to $\pi_\ell(r)\1$ (defined by
\eqref{eq:steady}) as $N$ tends to infinity. This result was established in \cite{minnebo2} for exponential job sizes.
A formal proof of this statement for phase-type distributed job sizes can be obtained by checking the conditions of Theorem 3.2 of 
\cite{gast2017expected}. While some of these conditions are not hard to verify, it is 
unclear at this stage how to prove that conditions (A2) and (A3) hold.
The following less general result follows from \cite{vanhoudt_global_PACM}: 
\begin{theo}
If we truncate all queues to length $B$ and assume the job sizes are hyperexponential,
then $\pi_\ell^{(N,B)}(r)$, which is the steady state probability that an arbitrary server has $\ell$ jobs in a system
consisting of $N$ servers with queues of length $B$ and probe rate $r$, converges to $\pi^{(B)}_\ell(r)\1$, which
is the steady state probability of the finite state QBD obtained by truncating $Q(r)$ such that $X_t(r) \leq B$.
\end{theo}
\begin{proof}
By means of Corollary 1 and the discussion in Section 7.2 of \cite{vanhoudt_global_PACM}, one can show 
that the sum $\sum_{\ell=k}^B \pi_\ell^{(N,B)}(r)$ converges to $\sum_{\ell=k}^B \pi^{(B)}_\ell(r)\1$
for any $k=1,\ldots,B$. 
\end{proof}
Thus convergence of the stationary regime of the sequence of Markov chains holds if we truncate the queues and limit
ourselves to hyperexponential job sizes. Truncating the queues mainly avoids a number of technical issues
(discussed in \cite{vanhoudt_global_PACM} after Corollary 1) and is a common practice in proving mean field convergence,
e.g., \cite{gast2010mean,ganesh2010}. From a practical point of view, there is also no difference between having a huge finite 
buffer or an infinite buffer (as long as the infinite system is stable). Relaxing the job size distribution to any phase-type distribution appears far less obvious. Although Corollary 1 in \cite{vanhoudt_global_PACM} allows us to consider 
a slightly broader class of jobs sizes than hyperexponential sizes, it is not possible to prove convergence 
in the same manner for any phase-type job size distribution.

%as one needs to define a norm on the state space such that the unique fixed point is a global attractor under this norm and 
%to identify a bounded subspace such that the $N$-dimensional process concentrates on this subspace at rate $O(1/N^2)$ in steady state.
%In fact, global attraction for the set of ODEs in (\ref{eq:ODEl>0}-\ref{eq:ODEl=0}) can be established in case
%of hyperexponential jobs sizes by relying on the methodology developed in \cite{vanhoudt_global_PACM} (see Section 4.2 in \cite{vanhoudt_global_PACM}). 

\section{Stealing versus Sharing}\label{sec:versus}

Our main objective is to compare the performance, i.e., mean response time of a job, under
stealing and sharing when the rate $r$ is set such that a predefined overall probe
rate $r_{overall}$ is matched by both strategies. Note that when both strategies rely on the
same $r$ value, we obtain the same queue length distribution and therefore the same mean response time.

For stealing it is clear that $r_{overall} = r(1-\lambda)$ as the idle servers probe at rate $r$.
Hence, there is a unique $r_{steal}$ that matches the predefined $r_{overall}$, that is, $r_{steal} = r_{overall}/(1-\lambda)$.
For work sharing the queues with pending jobs send probes at rate $r$, meaning 
$r_{overall}  = r \pi_{2+}(r)\1$ holds. Proposition \ref{prop:probe_rate_push} (illustrated in Figure \ref{fig:probe_rate}
and proven further on) implies that $r_{overall}  = r \pi_{2+}(r)\1$ has a unique solution $r_{share}$ for $r_{overall} \in [0,\lambda^2/(1-\lambda))$. 
For $r_{overall} \geq  \lambda^2/(1-\lambda)$, we have $r \pi_{2+}(r)\1 \leq r_{overall}$ for any $r \geq 0$ and 
the rate $r$ can thus be chosen arbitrarily large without exceeding $r_{overall}$. 

\begin{figure}[t]
\center
\includegraphics[width=0.48\textwidth]{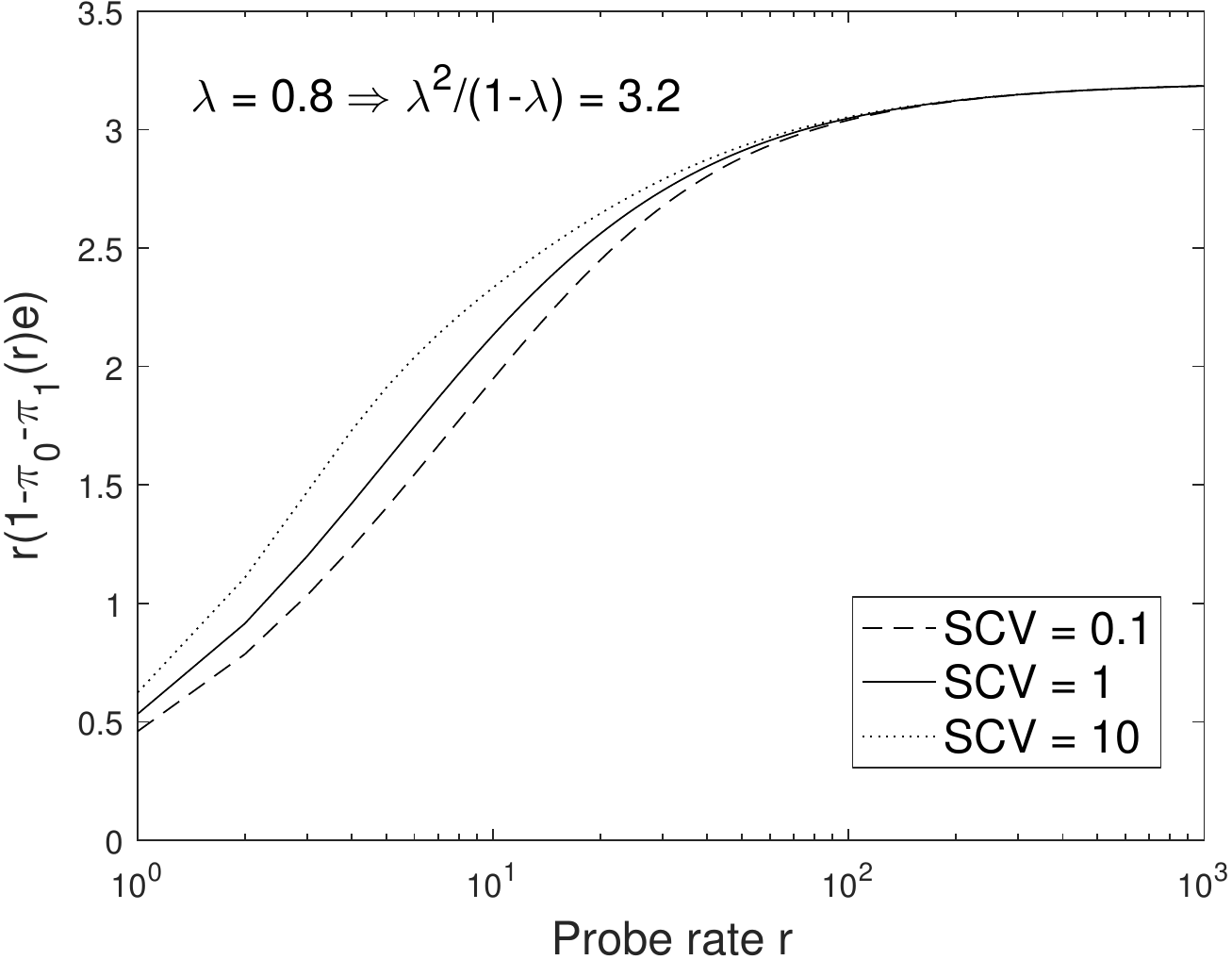}
\caption{Overall probe rate $r_{overall}$ as a function of probe rate $r$ for work sharing.}
\label{fig:probe_rate}
\end{figure}

In order to prove Proposition \ref{prop:probe_rate_push}, we first show that the mean response time decreases as 
the probe rate $r$ increases (via Little's law), as expected. 

\begin{prop}\label{prop:mresp}
The vector $\pi_{k+}(r)$, for $k \geq 2$, is decreasing in $r$ and so is the mean queue length.
\end{prop}
\begin{proof}
Consider the Markov chain $(\hat X_t(r),\hat Y_t(r))_{t \geq 0}$ defined by censoring the chain $(X_t(r),Y_t(r))_{t \geq 0}$ on
the states $(\ell,j)$ with $\ell > 0$ (i.e., when the queue is busy). Denote $\hat \pi_\ell(r)$ as its steady state probabilities. Clearly, 
$\pi_\ell(r) = \lambda \hat \pi_\ell(r)$ and it suffices to show that $\sum_{\ell \geq k} \hat \pi_\ell(r)$
decreases in $r$. 
Assume $r_1 < r_2$ and couple the Markov chains $(\hat X_t(r_1),\hat Y_t(r_1))_{t \geq 0}$ and $(\hat X_t(r_2),\hat Y_t(r_2))_{t \geq 0}$
as follows. As both servers are busy at all times we can couple the service such that $\hat Y_t(r_1) = \hat Y_t(r_2)$ and the service
completions in both chains occur at the same time. Similarly we can couple the rate $\lambda$ arrivals such that they occur simultaneously.
Whenever $\hat X_t(r_2), \hat X_t(r_1) > 1$ and $\hat X_t(r_2)$ decreases by one due to a job transfer 
(which happens at rate $(1-\lambda)r_2$), we decrease $\hat X_t(r_1)$ with probability $r_1/r_2$. 
If $\hat X_t(r_2)=1$ and $\hat X_t(r_1) > 1$, we decrease $\hat X_t(r_1) > 1$ by one at rate $(1-\lambda)r_1$ due to the job transfers.  
Therefore at any time $t$ we have $\hat Y_t(r_1) = \hat Y_t(r_2)$ and
$\hat X_t(r_1) \geq \hat X_t(r_2)$ as required.
The result for the mean queue length follows immediately as $\sum_{\ell \geq 1} \ell \pi_\ell(r) \1 = 
\sum_{k \geq 1} \pi_{k+}(r) \1$. 
\end{proof}

\begin{prop}\label{prop:probe_rate_push}
$r \pi_{2+}(r) \1$ increases as a function of $r$ and its limit as $r$ tends to infinity is $\lambda^2/(1-\lambda)$.
\end{prop}
\begin{proof}
Let $r_1 < r_2$ and consider the same coupled processes $(\hat X_t(r_1),\hat Y_t(r_1))_{t \geq 0}$ 
and $(\hat X_t(r_2),\hat Y_t(r_2))_{t \geq 0}$  defined in the proof of Proposition \ref{prop:mresp}. 
Let $\hat  Z_t(r_i)$ be the number of jobs transferred by process $(\hat X_t(r_i),\hat Y_t(r_i))_{t \geq 0}$ up to time $t$,
for $i=1,2$. We now claim that $\hat  Z_t(r_2) \geq \hat  Z_t(r_1)$ to prove the first part of the proposition
(as the stationary probe rate is equal to the stationary job transfer rate divided by $1-\lambda$).
We do this by arguing that $\hat Z_t(r_2)+\hat X_t(r_2) \geq \hat Z_t(r_1)+\hat X_t(r_1)$, which suffices 
as $\hat X_t(r_1) \geq \hat X_t(r_2)$.
Whenever $\hat X_t(r_2) > 1$ the sum $\hat Z_t(r_i)+\hat X_t(r_i)$, for $i=1,2$, increases (arrival) or decreases (job completion)
in the same manner for both processes (when a job transfer takes place the sum remains the same).  
If $\hat X_t(r_1) > 1$ and $\hat X_t(r_2) = 1$, the sum $\hat Z_t(r_1)+\hat X_t(r_1)$ again evolves as
$\hat Z_t(r_2)+\hat X_t(r_2)$, except that at service completions $\hat Z_t(r_1)+\hat X_t(r_1)$ decreases by one contrary to
$\hat Z_t(r_2)+\hat X_t(r_2)$ (which remains fixed as $\hat X_t(r_2)$ cannot decrease).

We now consider the limit of $r \pi_{2+}(r)\1$ as $r$ tends to infinity.
Assume $(X_0(r),Y_0(r))=(n+1,i)$ and let $T_n = \inf_t \{X_t(r) = n\}$, for $n > 1$, then
entry $(i,j)$ of the matrix $G(r)$ represents the probability that 
$Y_{T_n}(r)=j$ \cite{latouche1}. Therefore the limit as $r$ tends to infinity of
$G(r)$ is the identity matrix. Further as stated before, $\lambda G(r) = R(r) ((1-\lambda)r I + \mu \alpha)$.
As $R(r) \geq 0$, we have $\lim_{r \rightarrow \infty} R(r) = 0$,
$\lim_{r \rightarrow \infty} rR(r) = \lambda I/(1-\lambda) $ and
$\lim_{r \rightarrow \infty} r \sum_{\ell \geq 1} R(r)^\ell = \lambda I/(1-\lambda)$.
Hence, $ \lim_{r \rightarrow \infty} r \pi_{2+}(r)\1$ $= 
\lim_{r \rightarrow \infty}\pi_1(r) r \sum_{\ell \geq 1} R(r)^\ell \1 =
\lambda^2/(1-\lambda)$ as $\lim_{r \rightarrow \infty} \pi_1(r)\1 = \lambda$. 
\end{proof}

We can avoid computing the matching $r$ for work sharing when determining whether stealing or sharing performs 
best given the parameters $r_{overall}, \lambda, \alpha$ and $S$ using the next theorem.
\begin{theo}\label{th:comp}
Given $(\alpha,S)$, $\lambda$ and $r_{overall} > 0$, work sharing achieves a lower mean response time
than stealing if and only if
\begin{align}\label{eq:test}
1 - \lambda > \lambda - \pi_1(r_{overall}/(1-\lambda))\1 = \pi_{2+}(r_{overall}/(1-\lambda))\1.
\end{align}
\end{theo}
\begin{proof}
The probe rate $r$  that matches $r_{overall}$ for stealing is
$r_{steal} =r_{overall}/(1-\lambda)$. If work sharing uses the same $r$ its overall
probe rate, denoted as $r_{overall,share}$, would be
\[ r_{steal} \pi_{2+}(r_{steal})\1 = \frac{r_{overall}}{1-\lambda}\pi_{2+}(r_{overall}/(1-\lambda))\1.\] 
If and only if $r_{overall,share} < r_{overall}$, work sharing is using fewer probe messages
and due to Proposition \ref{prop:probe_rate_push} the unique solution $r_{share}$ to 
the equation $r_{overall} = r \pi_{2+}(r)\1$ is larger than $r_{steal}$.
If $r_{share} > r_{steal}$, then Proposition \ref{prop:mresp} implies that the mean queue length and
thus the mean response time of work sharing is less than the mean response time of work stealing. 
\end{proof}

In other words given $r_{overall}, \lambda, \alpha$ and $S$, we can determine whether stealing or sharing performs
best (in terms of the mean response time) by solving a single QBD-type Markov chain, which can be done 
numerically in a fraction of a second.

The result of Theorem \ref{th:comp} can also be intuitively understood using the following argument. Provided that
both strategies are using the same overall probe rate $r_{overall}$, the lowest mean
response time is achieved by the strategy for which a probe message has the highest probability of
resulting in a job transfer. For work sharing this probability is equal to $1-\lambda$,
the probability that a server is idle, while for work stealing it is given by the probability 
$\pi_{2+}(r_{overall}/(1-\lambda))\1$ that there are two or more jobs in a server (which depends on the
arrival rate, the job size distribution and the overall prove rate $r_{overall}$).
%Thus, \eqref{eq:test} simply states that the push strategy is better than the pull strategy 
%if the probability of having a successful probe is higher.

%We note that $\pi_1(r)e = (1-\lambda)\lambda_0(r) \alpha (-(S+\lambda(G(r)-I)))^{-1}e$ and 
%$-\alpha (S+\lambda(G(r)-I))^{-1}e$ is the mean time that the queue length equals $1$
%during a single busy period.

In the remainder of this section we establish two results:
\begin{enumerate}
\item Given $(\alpha,S)$ and $\lambda$ there exists a $r^*_{overall}$ such that
work sharing is best if and only if $r_{overall} > r^*_{overall}$.
\item Given $(\alpha,S)$ and $r_{overall}$ there exists a $\lambda^*$ such that
work sharing is best if and only if $\lambda < \lambda^*$.
\end{enumerate}
Moreover $r^*_{overall}$ and $\lambda^*$ can be easily computed using \eqref{eq:test} in combination
with a simple bisection algorithm.

\paragraph{Remark} By noting that 
\[\pi_1(r_{overall}/(1-\lambda)) = \lambda \left(1-\frac{\lambda}{1+r_{overall}}\right),\]
in case of exponential job sizes, one finds that \eqref{eq:test} simplifies to \eqref{eq:testexp}.

\begin{corol}\label{cor:r_overall}
Given  $(\alpha,S)$ and $\lambda > 0$, there exists an $r^*_{overall} \in [0,\lambda^2/(1-\lambda))$ such that 
work sharing achieves a lower mean response time than stealing if and only if
$r_{overall} > r^*_{overall}$.
\end{corol}
\begin{proof}
Due to Proposition \ref{prop:mresp} we have that  $\pi_{2+}(r)\1$ is decreasing in $r$. 
If $1-\lambda > \pi_{2+}(0)\1 $, \eqref{eq:test} implies that $r^*_{overall}=0$. 
Otherwise there exists a unique $r>0$ such that $1-\lambda = \pi_{2+}(r)\1$,
as $\pi_{2+}(r)\1$ tends to zero as $r$ tends to infinity.
Finally, when $r_{overall} \geq \lambda^2/(1-\lambda)$ work sharing can pick $r$ arbitrarily large
without exceeding $r_{overall}$, meaning the mean response time can be made arbitrarily close to one. 
\end{proof}

\begin{prop}\label{prop:inclambda}
The scalar $\pi_{k+}(r_{overall}/(1-\lambda))\1$, for $k\geq 1$ and $r_{overall}$ fixed, is increasing in $\lambda$.
\end{prop}
\begin{proof}
We make use of the following result \cite[Theorem 1]{busicMOR}.
Consider two stable discrete time Markov chains $\{\tilde X^{(1)}_n\}$ and $\{\tilde X^{(2)}_n\}$ on the state space $S=\{0,1,2,\ldots\}$ with 
steady state probabilities $\pi^{(1)}_i$ and $\pi^{(2)}_i$ for $i \in S$. Assume $\{\tilde X^{(2)}_n\}$ is obtained from $\{\tilde X^{(1)}_n\}$
by replacing some transitions from state $i$ to state $j$, by a transition from state $i$ to $j'$ with $j \leq j'$, then
$\sum_{\ell \geq k} \pi_\ell^{(1)} \leq \sum_{\ell \geq k} \pi_\ell^{(2)}$ for any $k$.

Let $\{X^{(1)}_t\}$ be the QBD Markov chain with arrival rate $\lambda^{(1)}$ and probe rate $r^{(1)} = r_{overall}/(1-\lambda^{(1)})$ and
 $\{X^{(2)}_t\}$ be the QBD Markov chain with arrival rate $\lambda^{(2)} > \lambda^{(1)}$ and probe rate $r^{(2)} = r_{overall}/(1-\lambda^{(2)})$. 
Denote their steady state probability vectors as $\pi^{(1)}_\ell(r^{(1)})$ and $\pi^{(2)}_\ell(r^{(2)})$, respectively.
Note that for $i=1,2$ we have $A_1^{(i)} = \lambda^{(i)} I$, $A_0^{(i)} = S-(\lambda^{(i)}+r_{overall})I$, $A_{-1}^{(i)}=\mu \alpha + r_{overall}I$.

Let $\{\hat X^{(i)}_t\}$ be the continuous time QBD Markov chain obtained from $\{X^{(i)}_t\}$ by 
censoring out the idle periods and let $\{\tilde X^{(i)}_n\}$ be the discrete time Markov chain 
obtained from $\{\hat X^{(i)}_t\}$ by uniformization. Denote their respective steady state probabilities
as  $\hat \pi^{(i)}_\ell(r^{(i)})$ and $\tilde \pi^{(i)}_\ell(r^{(i)})$. 
Due to \cite[Theorem 1]{busicMOR} we now have  
$\sum_{\ell \geq k} \tilde \pi_\ell^{(1)}(r^{(1)}) \leq \sum_{\ell \geq k} \tilde \pi_\ell^{(2)}(r^{(2)})$ for any $k \geq 1$. 
Further, $ \pi_\ell^{(i)}(r^{(i)}) = \lambda^{(i)} \hat \pi_\ell^{(i)}(r^{(i)}) = \lambda^{(i)} \tilde \pi_\ell^{(i)}(r^{(i)})$ for
$\ell \geq 1$ by construction (as the rate of exiting level $0$ is exactly such that
$\lambda^{(i)}$ is the probability that the server is busy). This yields 
$\sum_{\ell \geq k} \pi_\ell^{(1)}(r^{(1)}) = \sum_{\ell \geq k} \lambda^{(1)} \tilde \pi_\ell^{(1)}(r^{(1)})
< \sum_{\ell \geq k} \lambda^{(2)} \tilde \pi_\ell^{(1)}(r^{(1)}) \leq \sum_{\ell \geq k} \lambda^{(2)} \tilde \pi_\ell^{(2)}(r^{(2)})
= \sum_{\ell \geq k} \pi_\ell^{(2)}(r^{(2)})$ for $k \geq 1$ as required.
\end{proof}

\begin{theo}\label{th:lambda}
Given  $(\alpha,S)$ and $r_{overall} > 0$, there exists a $\lambda^* \in (0,1)$ such that work sharing 
is best if and only if $\lambda < \lambda^*$. 
\end{theo}
\begin{proof}
By Proposition \ref{prop:inclambda} the scalar $\pi_{2+}(r_{overall}/(1-\lambda))\1$ increases in $\lambda$,
hence there exists a unique solution in $(0,1)$ such that $1-\lambda = \pi_{2+}(r_{overall}/(1-\lambda))\1$, which implies the
result due to \eqref{eq:test}.
\end{proof}

\begin{table}[t]
\center
\begin{tabular}{c c c c c}
 \hline
$SCV$ & $f$ & $p_1$ & $1/\mu_1$ & $1/\mu_2$ \\ \hline
1 & 1/2& 0.5000 & 1 & 1 \\
5 & 1/2&0.9082  & 0.5505  &  5.4495 \\
25 & 1/2&0.9804 & 0.5100 & 25.4900  \\
10 & 1/10& 0.8524 & 0.1173 & 6.0981 \\
100 & 1/100& 0.9806 & 0.0102 & 51.0100 \\
1000 & 1/1000& 0.9980 & 0.0010 & 501.0010 \\
\hline
\end{tabular}
\caption{Parameter settings of the hyperexponential job size distribution for
various $SCV$ and $f$ values.}
\label{tab:HEXP}  
\end{table}

\section{Numerical results}\label{sec:num}

In Figure \ref{fig:f1/2} we present the boundary between the regions where work stealing and sharing result in the
lowest mean response time for $5$ different job size distributions ($2$ hyperexponential, the exponential and
$2$ Erlang distributions). This figure illustrates that the region
where stealing is best grows as the SCV of the job size distribution increases. Thus, stealing benefits from having more variable job sizes. 
This can be understood as follows. When job sizes become 
more variable, the probability of having two or more jobs in a queue tends to increase. When $r_{overall}$ is
fixed both work sharing and stealing use the same number of probes per time unit. However, for work sharing
probes are successful with probability $1-\lambda$, irrespective of the job size variability. For stealing 
the success probability equals the probability of find two or more jobs and thus increases with the
job size variability. In short, for $r_{overall}$ fixed,  as the job sizes become more variable, the number of job exchanges increases
for work stealing, while it remains the same for work sharing. Therefore, the region where stealing is
superior grows as the job size distribution becomes more variable.

Figure \ref{fig:SCV10} illustrates that the boundary
between stealing and sharing is not fully determined by the first two moments of the job size
distribution. This is not unexpected as the probability to have two or more jobs in an M/PH/1 queue
(with or without negative arrivals) also depends on the higher moments for the job sizes. 
 We do note that the regions where work sharing is best for the $5$ distributions with SCV equal to 10 are
smaller compared to exponential job sizes.

\begin{figure}[t]
\center
\includegraphics[width=0.48\textwidth]{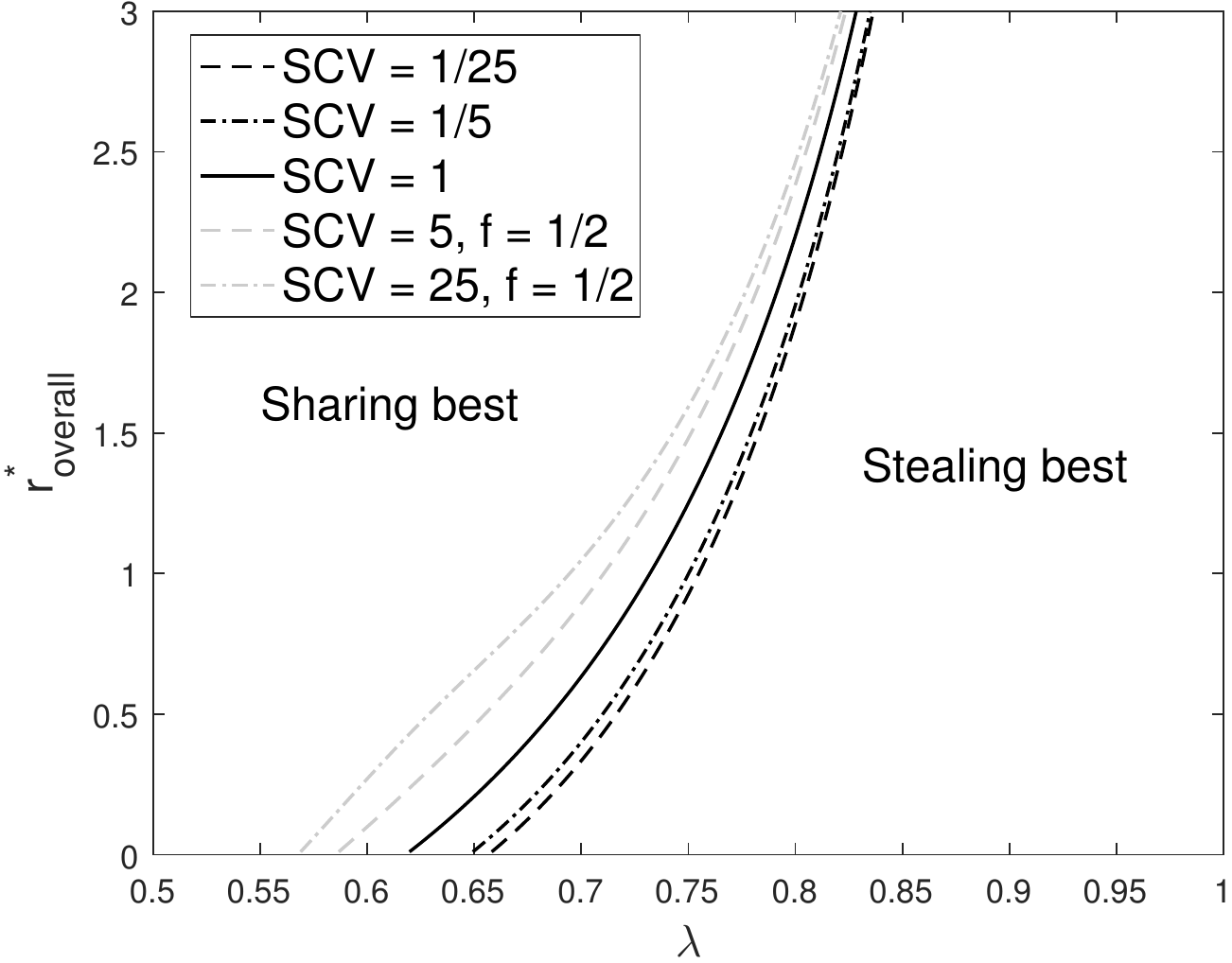}
\caption{Region for which stealing/sharing achieves the lowest mean response time.}
\label{fig:f1/2}
\end{figure}

\begin{figure}[t]
\center
\includegraphics[width=0.48\textwidth]{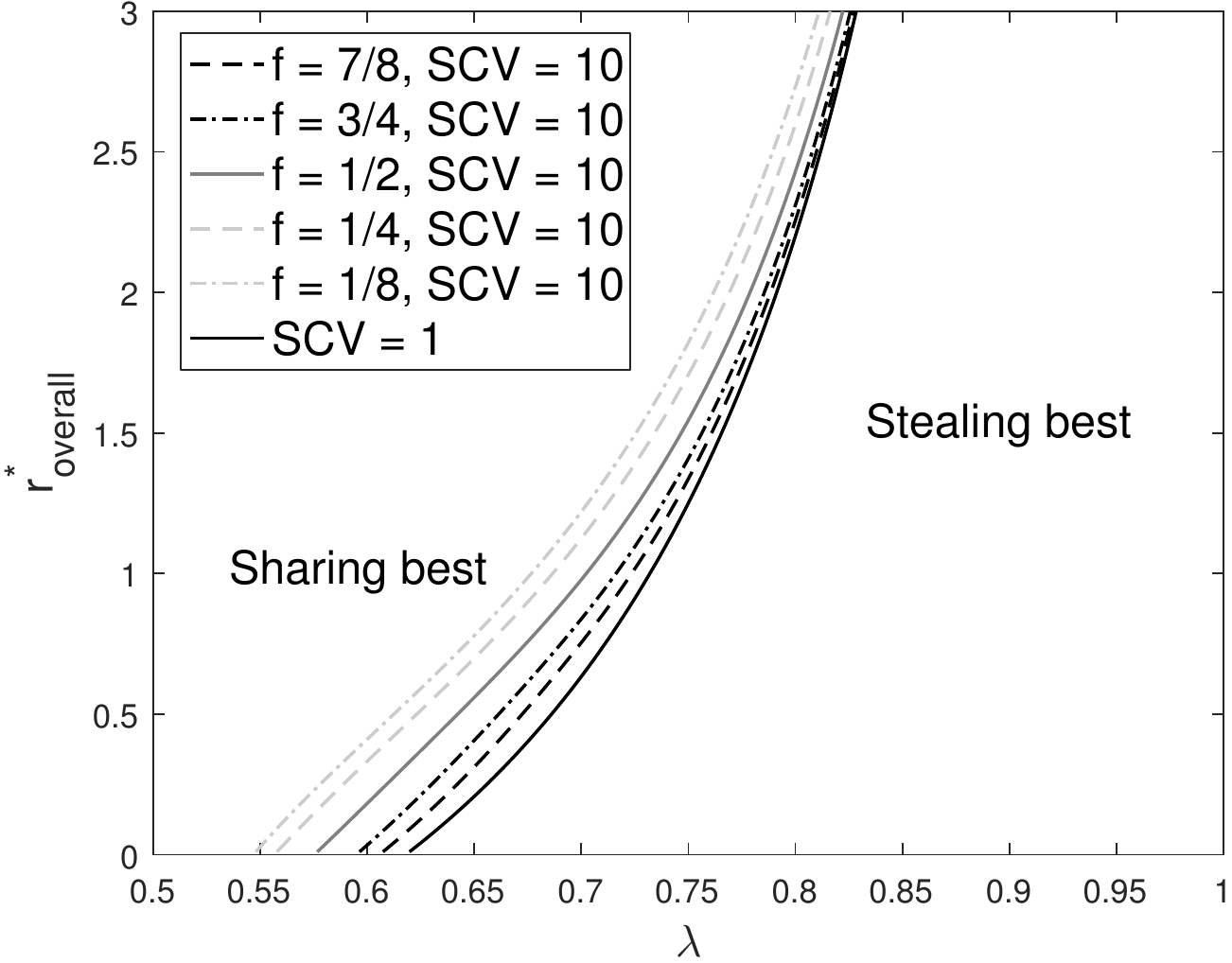}
\caption{Region for which stealing/sharing  achieves the lowest mean response time.}
\label{fig:SCV10}
\end{figure}

These figures trigger a number of questions:
\begin{enumerate}
\item Can we identify a region where stealing/sharing is the best for any (phase-type) job size distribution? In other
words, how far to the right/left can the stealing/sharing boundary move?
\item Given some information on the job size distribution, e.g., if the job size distribution has an increasing/decreasing hazard rate,
can we identify a more narrow region that contains the stealing/sharing boundary?
\item Is it possible to explicitly characterize the boundary  in some cases, e.g., as $r_{overall}$ tends to zero?
\end{enumerate}
In the next section we focus on the first two questions, while the third question is considered Section \ref{sec:small}.

\section{Bounds}\label{sec:bounds}

Although we can easily compute the boundary between the regions on a $(\lambda,r_{overall})$ plot where stealing/sharing 
prevails for a specific phase-type distribution $(\alpha,S)$, we now aim
at establishing simple bounds on these regions that are valid for any phase-type distribution. 

\subsection{General bounds}

We start with a tight general work sharing bound:

\begin{corol}[General Work Sharing Bound]\label{cor:pushbound}
When the job sizes follow a phase-type distribution $(\alpha,S)$ (with mean $1$), 
work sharing achieves a lower mean response time than stealing if
\begin{align}\label{eq:bound}
\lambda < \frac{\max(1,\sqrt{r_{overall}(r_{overall}+4)}-r_{overall})}{2}.
\end{align}
\end{corol}
\begin{proof}
The $1/2$ bound is immediate from Theorem \ref{th:comp} as $\pi_1(r)\1 > 0$ for any $r \geq 0$.
The $(\sqrt{r_{overall}(r_{overall}+4)}-r_{overall})/2$ bound follows from the fact that work sharing has a mean response time 
arbitrarily close to $1$ when $r_{overall}$ exceeds $\lambda^2/(1-\lambda)$ as $r$ can be set arbitrarily large without
exceeding $r_{overall}$.
\end{proof}

In Figure \ref{fig:tight} we numerically illustrate that the bound in \eqref{eq:bound} is tight. 
More specifically, this bound is approached when a large majority of the jobs is very short (that is, $p_1 \approx 1$ and $1/\mu_1 \approx 0$)
and the remaining fraction of the jobs contributes nearly the entire workload (i.e., $f \approx 0$).

Note that \eqref{eq:bound} is identical to 
the exponential bound given by \eqref{eq:testexp} if we replace $r_{overall}$
by $r_{overall}+1$.

\begin{figure}[t]
\center
\includegraphics[width=0.48\textwidth]{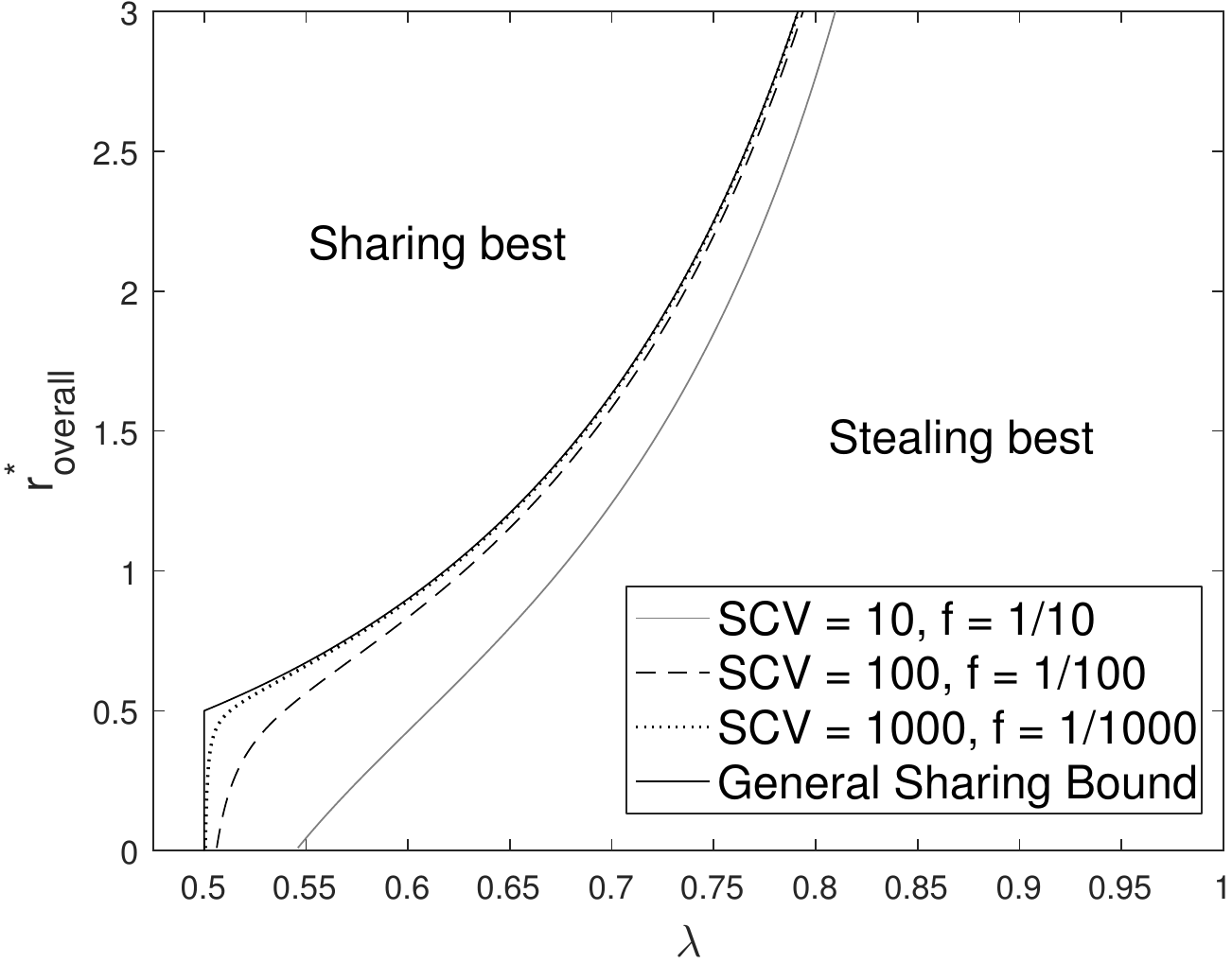}
\caption{Numerical illustration of the tightness of the general work sharing bound.}
\label{fig:tight}
\end{figure}

For the general work stealing bound we have the following conjecture. The idea behind this conjecture is
that we believe that the probability to have two or more jobs is minimized over all job size distributions 
with mean $1$ by the deterministic distribution in an M/G/1 queue with negative customers.
While this may appear as a rather intuitive result, we did not manage to come up with a formal proof thus far.
The difficulty is due to the fact that the negative arrivals remove customers from the back of the queue.

\begin{conj}[General Work Stealing Bound]\label{conj:pullbound}
When the job sizes follow a phase-type distribution $(\alpha,S)$ (with mean $1$), stealing 
achieves a lower mean response time than work sharing if $\lambda > \nu(r_{overall})$,
where $\nu(r_{overall})$ is the unique solution on $(0,1)$ of
\begin{align}
1-\lambda = \pi_{2+}^{(det)}(r_{overall}/(1-\lambda)),
\end{align}
where $\pi_{2+}^{(det)}(r)$ is the probability that the queue length exceeds one 
in the M/PH/1 queue with negative arrivals characterized by \eqref{eq:Q} when the phase-type job sizes
are replaced by deterministic job sizes.
\end{conj}
We note that the probability $\pi_{2+}^{(det)}(r)$ is not easy to compute due to the negative arrivals.

Proving the above conjecture can be shown to be equivalent to proving
the following statement: Consider an M/M/1 queue with negative arrivals where the arrival rate
$\lambda < 1$, the service rate $r_{overall} \geq 0$ and the negative arrivals are generated by
an independent renewal process with a mean inter-renewal time equal to one, then the probability of
having an empty queue is maximized over all renewal processes by the deterministic
renewal process.

\subsection{Monotone hazard rate bounds}

Given a continuous distribution with cdf $H$ and pdf $h$, the hazard rate $z(t)$ is defined
as $h(t)/(1-H(t))$, which in case of a phase-type distribution $(\alpha,S)$ means $z(t)=\alpha e^{St} \mu/\alpha e^{St} \1$.
For distributions with a monotone hazard rate we conjecture the following:
\begin{conj}[DHR/IHR bounds]\label{conj:DHR/IHR}
The exponential bound specified by
\eqref{eq:testexp} corresponds to a sharing (stealing) bound for any phase-type job size distribution with
a decreasing (increasing) hazard rate.
\end{conj}
In other words the boundary between stealing and sharing 
region for a job size distribution with a decreasing (increasing) hazard rate
is located between the exponential boundary and the general stealing (sharing) boundary
as illustrated in Figure \ref{fig:conj}.

\begin{figure}[t]
\center
\includegraphics[width=0.48\textwidth]{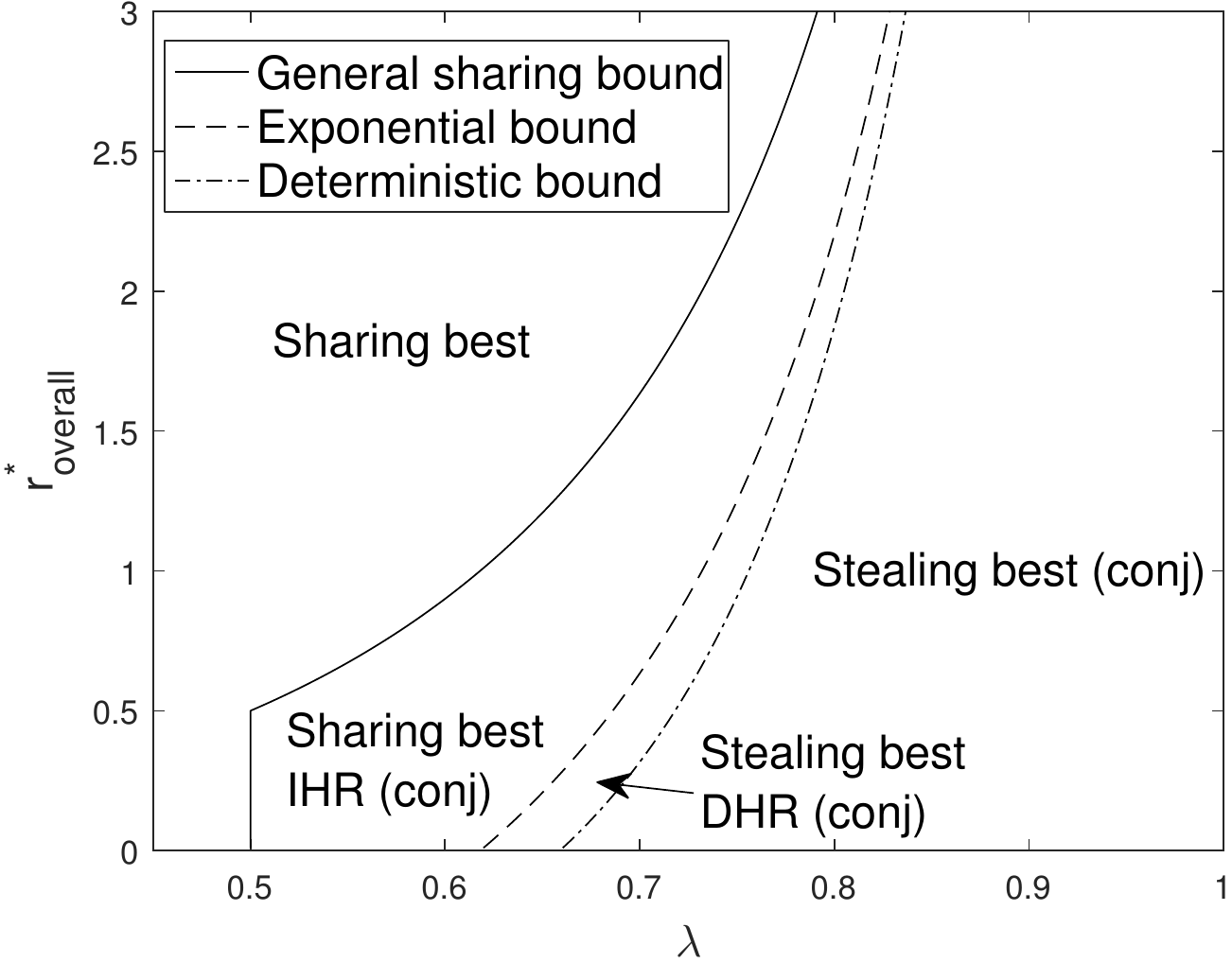}
\caption{Illustration of Conjectures \ref{conj:pullbound} and \ref{conj:DHR/IHR}.}
\label{fig:conj}
\end{figure}

At this stage we do not have a proof for Conjecture 1 and 2. In the next section we show that these
conjectures are valid when the probe rate tends to zero. In the remainder of this section we establish
weaker bounds for job sizes with increasing/decreasing hazard rates. 

\begin{figure}[t]
\center
\includegraphics[width=0.48\textwidth]{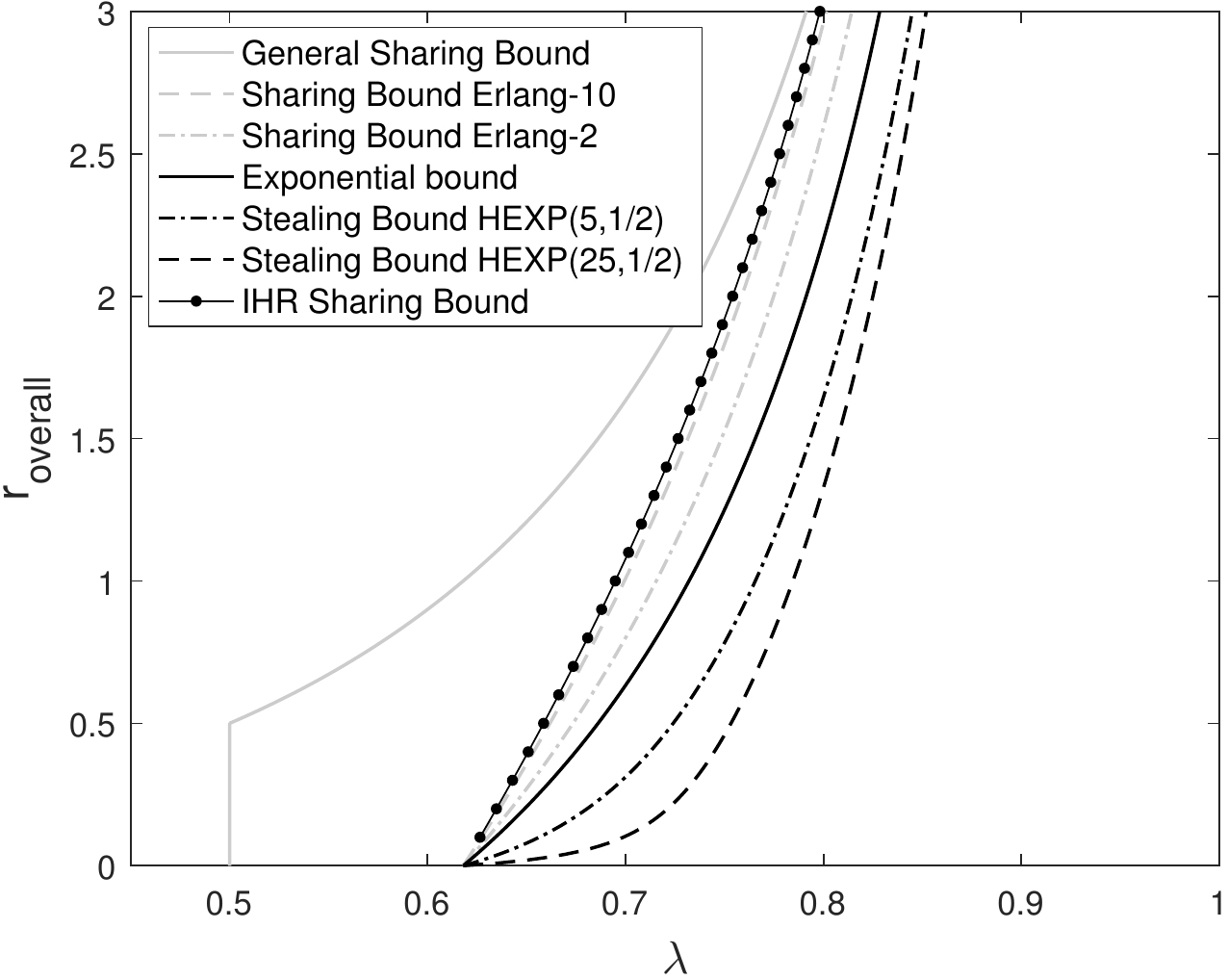}
\caption{Illustrated stealing and sharing bounds.}
\label{fig:hazard}
\end{figure}

\begin{prop}[DHR Stealing Bound, IHR sharing Bound]\label{prop:boundhr}
Let $Y$ be the minimum of the phase-type job size distribution and an exponential random variable with parameter $r_{overall}$,
that is, $E[Y] = -\alpha (S-r_{overall}I)^{-1}\1$. Define 
\begin{align}\label{eq:L}
L(x)=\frac{\sqrt{1+4x}-1}{2x}.
\end{align}
When the job size has a decreasing hazard rate (DHR), stealing achieves a lower mean response time than sharing 
if $\lambda > L(E[Y])$.
When the job size has a increasing hazard rate (IHR), sharing achieves a lower mean response time than stealing if
$\lambda < L(E[Y])$.
\end{prop}
\begin{proof}
Consider the queueing system corresponding to the Markov chain with rate matrix $Q(r_{overall}/(1-\lambda))$, see
\eqref{eq:Q}. Note that in this
queueing system the job transfers are also regarded as job completions and consecutive service times are therefore correlated. 
Let $X$ be the queue length and $R_i$, for $i>0$, be the amount of time that a job spends waiting at the {\it head} of the
waiting room provided that the job arrived when the queue length equaled $i$. 
Assume we collect reward at rate $1$ when the queue length exceeds $1$, thus the average rate at which we collect reward
is $P[X \geq 2]$. This average reward should be equal to the rate of customer arrivals that generate reward times the average reward
that each such customer delivers, thus
\[ P[X \geq 2] = \lambda \sum_{i\geq 1} P[X = i] E[R_i],\]
as the arrival rate is $\lambda$ (unless the queue is idle). In general the difficulty lies in bounding $E[R_i]$.
However, for decreasing hazard rate job sizes\footnote{If X is DHR/IHR, then so is $\min(X,Y)$ if
$Y$ is exponential, as $z_{\min(X,Y)}(t) = z_X(t)+z_Y(t)$.} the expected time that a job stays at the head of the waiting room is lower bounded
 by assuming that the job in service just started service (and was not caused by a transfer).
This implies that $E[R_i] \geq E[Y]$, where $Y$ is defined as the minimum of the phase-type job length
and an exponential random variable with parameter $r_{overall}$.  
Hence, 
\[P[X \geq 2] \geq \lambda E[Y] P[X \geq 1] = \lambda^2 E[Y].\] 
The result now follows from \eqref{eq:test}, which implies that stealing is best if $1-\lambda < \lambda^2 E[Y]$.

The argument for the increasing hazard rate case is identical, except that 
$E[R_i]$ is now upper bounded by $E[Y]$. Therefore sharing is best if $1-\lambda > \lambda^2 E[Y]$
\end{proof}

Note that when the job sizes are hyperexponential, so is $Y$ and $E[Y] = \sum_{i=1}^k \frac{p_i}{\mu_i + r_{overall}}$.
The above bounds are tight in case of exponential job sizes only as illustrated in Figure \ref{fig:hazard}.

\begin{corol}[General IHR Sharing bound]
For phase-type job sizes with increasing hazard rate work sharing achieves a lower mean response time than stealing if
\[\lambda < L((1-e^{-r_{overall}})/r_{overall}),\] 
where $L(x)$ is defined by \eqref{eq:L}.
\end{corol}
\begin{proof}
As $L(x)$ is decreasing in $x$ on $[0,1]$, the result is immediate from Proposition \ref{prop:boundhr} provided that $E[Y] \leq (1-e^{-r_{overall}})/r_{overall}$.
Let $F(x)$ be the CDF of the (phase-type) job size distribution $X$, then
\begin{align*}
E[Y] &= \int_{x=0}^\infty (1-F(x))e^{-r_{overall}x} dx 
\\&= \frac{1}{r_{overall}} (1-E[e^{-r_{overall}X}]).
\end{align*}
By Jensen's inequality $E[e^{-r_{overall}X}] \geq e^{-r_{overall}E[X]}$ with $E[X]=1$, which yields the required upper bound on $E[Y]$.
\end{proof}

We cannot obtain a meaningful general DHR stealing bound in the same manner as $E[Y]$ can be made arbitrarily small, which
yields a stealing bound $\lambda > 1$ (as $L(x)$ approaches $1$ as $x$ tends to zero).

\section{Small probe rates}\label{sec:small}
In this section we characterize the boundary between the stealing and sharing region for 
$r_{overall}$ sufficiently small. In order to do so we first show that the steady state vector
$\pi_\ell(r)$ is continuous in $r$ on $[0,\infty)$, meaning we can study
$\pi_\ell(r)$ for $r$ small by looking at the limit with $r=0$.

In order to establish the continuity we recall the following result: 
\begin{prop}[due to Corollary 3.9.1 of \cite{neuts2}]\label{prop:Neuts}
Let $D_1$ be a matrix with negative diagonal elements, non-negative off-diagonal elements and assume 
$D_1^{-1}$ exists. Let $D_2$ be a non-negative matrix such that $(D_1+D_2)\1=0$. Let $R$ be a 
non-negative matrix with spectral radius $0 < sp(R) < 1$, then
\[I \otimes D_1 + R^T \otimes D_2\] 
is non-singular, where $\otimes$ denotes the Kronecker product and $R^T$ the transposed
matrix of $R$.  
\end{prop}

\begin{theo}
The vector $\pi_\ell(r)$, for $\ell \geq 1$, is continuous in $r$ on $[0,\infty)$.
\end{theo}
\begin{proof}
We show that the matrix $R(r)$ is continuous on $[0,\infty)$, from which the continuity of $R_1(r)$ and
$\pi_\ell(r)$ follow. Consider the map $f$ from $\mathbb{R}^{m^2+1}$ to $\mathbb{R}^{m^2}$ that maps
$(X,r)$, where $X$ is a square matrix of size $m$, to $A_1+XA_0(r)+X^2A_{-1}(r)$. Note that
$(R(r),r)$ is mapped to zero by $f$. Let $J_f$ be the Jacobian of $f$ such that
$J_f(X,r) =[Y(X,r) | Z(X,r)]$ where $Y(X,r)$ is a square matrix of size $m^2$ with entry $((i,j),(i',j'))$ equal
to the partial derivative $\partial (A_1+XA_0(r)+X^2A_{-1})_{i,j}/\partial x_{i',j'}$. 
It is easy to verify that 
\[Y(X,r) = (I \otimes (A_0(r) + X A_{-1}(r))^T) + (X \otimes A_{-1}(r)^T).\]
As $(R(r),r)$ is a zero of $f$, the implicit function theorem states that there exists an open set $U$ containing $r$ 
such that there exists a unique continuously differentiable function $g$ from $U$ to $\mathbb{R}^{m^2}$ such that 
$f(g(r'),r')=0$ for any $r' \in U$ provided that $Y(R(r),r)$ is non-singular.
Now, as $R(r) A_{-1}(r)  = \lambda G(r)$ (see Section \ref{sec:MPH1}), we have
\begin{align*}
Y&(R(r),r) = \\
&(I \otimes (A_0(r) + R(r)A_{-1}(r))^T) + (R(r) \otimes A_{-1}(r)^T) =\\
&(I \otimes (A_0(r) + \lambda G(r))^T) + (R(r) \otimes A_{-1}(r)^T), 
\end{align*}
where $R(r)$ is a non-negative matrix with spectral radius $sp(R(r))\in (0,1)$,
$A_0(r) + \lambda G(r) = (S - (1-\lambda)r I+ \lambda(G(r)-I))$ has negative diagonal elements, non-negative
off-diagonal elements and is invertible (as it is a subgenerator matrix), while
$A_{-1}(r) = \mu \alpha +(1-\lambda)r I$ is a non-negative matrix and $(A_0(r) + \lambda G(r) + A_{-1}(r))\1 =
(\mu \alpha + S)\1 + \lambda(G(r)-I)\1 = 0$
as $G(r)\1 = \1$ and $\mu \alpha +S$ is a generator matrix. The non-singularity of $Y(R(r),r)$ therefore follows from
Proposition \ref{prop:Neuts}. 
\end{proof}

The following basic result on the M/G/1 queue is used in combination with the continuity and Theorem \ref{th:comp} to describe
the stealing/sharing boundary for $r_{overall}$ sufficiently small:

\begin{prop}\label{lemma:MG1_pi1}
The probability $\pi_{1-}$ that the queue length is at most one in an M/G/1 queue with arrival rate $\lambda$
and mean service time $1$ is given by $(1-\lambda)/g(\lambda)$, where $g(s)$ is the 
Laplace transform of the service time. Further $\pi_{1-} \leq (1-\lambda)e^\lambda$
(with equality for deterministic service).
\end{prop}
\begin{proof}
Let $\pi(z)$ be the generating function of the queue length distribution of an M/G/1 queue with mean service time $1$.
The Pollaczek-Khinchin formula states that
\[\pi(z) = \frac{(1-z)(1-\lambda)g(\lambda(1-z))}{g(\lambda(1-z))-z}.\]
The first result can be obtained by evaluating the derivative of $\pi(z)$ in $z=0$.
The inequality $e^{-\lambda} \leq g(\lambda)$ follows from Jensen's inequality (as it implies that $e^{E[X]} \leq E[e^X]$
for any random variable $X$).
\end{proof}

We can now show that Conjectures \ref{conj:pullbound} and \ref{conj:DHR/IHR} hold for $r$ tending to zero.  

\begin{prop}\label{prop:small}
For $r$ tending to zero stealing is best if  $\lambda >\nu \approx 0.6589$, where
$\nu$ is the unique solution of $\lambda/(1-\lambda) = e^\lambda$ in $(0,1)$.
\end{prop}
\begin{proof}
Proposition \ref{lemma:MG1_pi1} implies that as $r$ tends to zero, $\pi_{2+}(r)$ is lower bounded by
$1 - (1-\lambda)e^\lambda $. The unique solution of \eqref{eq:test} is therefore smaller
than the unique solution $\nu$ to $1-\lambda = 1 - (1-\lambda)e^\lambda $.
\end{proof}

\begin{prop}\label{prop:dhr}
For $r$ tending to zero and a phase-type job size distribution with decreasing (increasing) hazard rate, 
stealing (sharing) is best if  $\lambda >\phi-1$ ($\lambda < \phi-1$), where
$\phi$ is the golden ratio.
\end{prop}
\begin{proof}
The result is immediate from Proposition \ref{prop:boundhr} as $E[Y]$ tends to $1$ as $r_{overall}$ tends to zero.
\end{proof}

We end this section by characterizing the limit of the stealing/sharing boundary when $r$ tends to zero for
Erlang, hypoexponential and hyperexponential distributions.

\begin{prop}\label{prop:smallErlang}
For $r$ tending to zero and Erlang-$k$ job sizes with mean one, 
sharing is best if and only if $\lambda < \lambda_k^*$, where $\lambda_k^*$ 
is the unique solution of $\lambda/(1-\lambda) = (1+\lambda/k)^k$ in $(0,1)$. 
Further, the sequence $(\lambda_k^*)_k$ increases to the unique solution $\nu \approx 0.6589$ of
$\lambda/(1-\lambda) = e^\lambda$.
\end{prop}
\begin{proof}
As the Laplace transform of the Erlang-$k$ distribution with mean $1$ is given by $
(k/(k+\lambda))^k =(1+\lambda/k)^{-k}$,
Proposition \ref{lemma:MG1_pi1} and the continuity imply that 
\[\lim_{r \rightarrow 0} (1-\pi_{2+}(r)\1) = (1-\lambda)\left(1+\frac{\lambda}{k}\right)^k.\]
As such \eqref{eq:test} indicates that sharing is best, for $r$ tending to zero,
if and only if
\[\frac{\lambda}{1-\lambda} < \left(1+\frac{\lambda}{k}\right)^k.\]
Hence $\lambda_k^*$ is the unique solution in $(0,1)$ of $\lambda/(1-\lambda)=(1+\lambda/k)^k$. Further,
$\frac{d}{dk}(1+\lambda/k)^k > 0$ for $\lambda \in (0,1)$ and $\lim_{k \rightarrow \infty} (1+\lambda/k)^k  = e^\lambda$.
\end{proof}

\begin{prop}
For $r$ tending to zero and hypoexponential job sizes with $k$ phases and mean one, 
sharing is best if and only if $\lambda < \lambda^*$, where $\lambda^*$ 
is the unique positive solution of $\lambda/(1-\lambda) = \prod_{i=1}^k (1+\lambda/\mu_i)$ in $(0,1)$. 
Further, $\phi-1 \leq \lambda^* \leq \lambda^*_k$, where $\phi$ is the golden ratio and $\lambda_k^*$ is
defined in Proposition \ref{prop:smallErlang}.
\end{prop}
\begin{proof}
The first part of proof is identical to Proposition \ref{prop:smallErlang} except that
\[1/g(\lambda)  = \prod_{i=1}^k \left(1+\frac{\lambda}{\mu_i}\right).\]
The second part follows by showing that $(1+\lambda) \leq \prod_{i=1}^k (1+\lambda/\mu_i) \leq
\left(1+\frac{\lambda}{k}\right)^k$ for any $\mu_i \geq 0$ such that
$\sum_{i=1}^k 1/\mu_i = 1$.  The first inequality is immediate.
By defining $x_i = \lambda/\mu_i +1$, the second inequality follows from
the fact that $\prod_{i=1}^k x_i$ with $\sum_{i=1}^k x_i = \lambda + k$
is maximized by setting $x_i = (\lambda+k)/k$ (i.e., $\mu_i = k$).
\end{proof}

\begin{prop}\label{prop:smallhyp}
For $r$ tending to zero and hyperexponential job sizes with $k$ phases and mean one, 
sharing is best if and only if $\lambda < \lambda^*$, where 
$\lambda^*$ is the unique solution of $1/\lambda = 1 + \sum_{i=1}^k p_i \mu_i /(\lambda+\mu_i)$ on $(0,1)$.
Further $\lambda^* \leq \phi -1$ with $\phi$ the golden ratio.
\end{prop}
\begin{proof}
As $g(\lambda) = \sum_{i=1}^k p_i \mu_i/(\lambda+\mu_i)$, Proposition \ref{lemma:MG1_pi1} and the continuity imply
\begin{align*}
\lim_{r \rightarrow 0} (1-\pi_{2+}(r)\1) &= \frac{1-\lambda}{\sum_{i=1}^k p_i \mu_i/(\lambda+\mu_i)}.
\end{align*}
Thus, for $r$ tending to zero, sharing is best if and only if
\[\frac{\lambda}{1-\lambda} < \left( \sum_{i=1}^k p_i \mu_i/(\lambda+\mu_i) \right)^{-1}.\]
To establish the upper bound on $\lambda^*$ we need to show that the right hand side is bounded by $(1+\lambda)$.
Using the finite form of Jensen's inequality with $\phi(x)=1/x$, we get
\begin{align*} 
\phi(\sum_{i=1}^k p_i \mu_i/(\lambda+\mu_i)) &\leq \sum_{i=1}^k p_i \phi(\mu_i/(\lambda+\mu_i)) \\
&= \lambda \sum_{i=1}^k p_i/\mu_i + \sum_{i=1}^k p_i = \lambda +1,
\end{align*}
as the mean job length equals one.
\end{proof}

\section{Conclusions and future work}\label{sec:concl}
We introduced a mean field model for work sharing and work stealing with phase-type distributed job sizes and
indicated how to determine whether sharing or stealing is best for a given arrival rate, 
overall probe rate and job size distribution. Bounds that apply to any phase-type job size distribution
on the region where sharing/stealing is best were also discussed.
The main insight is that work stealing benefits considerably as the job sizes become more variable and may be
superior to work sharing for loads only marginally exceeding $1/2$ for some workloads.

The work sharing strategy considered in this paper is such that a server sends probe messages at a fixed rate $r$ whenever
it has pending jobs. One could also consider more advanced schemes where probing only starts if the number of pending jobs
exceeds some threshold or, more generally, where the probe rate depends on the queue length. Similarly one can also consider 
using a second (smaller) threshold and allowing servers with fewer jobs to accept job transfers. Such strategies were considered in
\cite{minnebo3},\cite{minnebo5} in case of exponential jobs sizes. We expect that the approach presented in this paper is also 
applicable to analyze such strategies with non-exponential job sizes.

Future work may exist in showing that the unique fixed point of the mean field model corresponds to the limit of the 
finite dimensional stationary distributions as well as proving the conjectured general work sharing bound. 
Other extensions might exist in studying these (or other) work sharing and stealing strategies in combination with load
balancing schemes such as Join-the-Shortest-Queue among $d$ randomly selected servers or Join-the-Idle-Queue.

\section*{Acknowledgment}
The author likes to thank Kostia Avrachenkov for a useful discussion on proving
the continuity of $\pi_\ell(r)$.

\bibliographystyle{plain}
\bibliography{../PhD/thesis}

%\begin{IEEEbiography}[{\includegraphics[width=1in,height=1.25in,clip,keepaspectratio]{benny.jpg}}]{Benny Van Houdt (benny.vanhoudt@uantwerpen.be) received his M.Sc. degree in Mathematics and Computer Science, and a PhD in Science from the University of Antwerp (Belgium) in July 1997, and May 2001, respectively.  He has been a postdoctoral fellow of the FWO-Flanders from 2001 until 2007. In 2007, he became a professor at the Mathematics and Computer Science Department at the University of Antwerp. His main research interest goes to the performance evaluation and stochastic modeling of computer systems and communication networks. He is currently the editor-in-chief of the Performance Evaluation journal and has received several awards, including best paper awards at ACM Sigmetrics and IFIP Performance.}\end{IEEEbiography}

\end{document}